\documentclass[a4paper,11pt]{article}

\usepackage{graphicx}
\usepackage{fullpage}
\usepackage{libertine}
\usepackage{color}
\usepackage[ruled,noline,noend,linesnumbered]{algorithm2e}
\SetArgSty{textrm} 

\usepackage{amsmath,amsfonts,amssymb,amsthm}

\usepackage[breaklinks=true]{hyperref}
\usepackage[svgnames,dvipsnames]{xcolor}
\usepackage[capitalise,nameinlink]{cleveref}
\hypersetup{colorlinks={true},linkcolor={DarkBlue},citecolor=[named]{DarkGreen}}

\crefname{mechanism}{Mechanism}{Mechanisms}

\crefformat{alg2}{#2RemoveForwardEdges#3}

\usepackage{natbib}

\usepackage{subcaption}

\usepackage{tikz}  
\usetikzlibrary{arrows}
\usetikzlibrary{patterns,snakes}
\usetikzlibrary{decorations.shapes}
\tikzstyle{overbrace text style}=[font=\tiny, above, pos=.5, yshift=5pt]
\tikzstyle{overbrace style}=[decorate,decoration={brace,raise=5pt,amplitude=3pt}]
\usetikzlibrary{shapes.geometric}

\newtheorem{inftheorem}{Informal Theorem}
\newtheorem{remark}{Remark}

\title{\bf Optimal Metric Distortion for Matching on the Line}

\usepackage{authblk}
\author[1]{Aris Filos-Ratsikas}
\author[2]{Vasilis Gkatzelis} 
\author[1]{Mohamad Latifian}
\author[2]{\\Emma Rewinski}
\author[3]{Alexandros A. Voudouris}

\affil[1]{University of Edinburgh, UK}
\affil[2]{Drexel University, USA}
\affil[3]{University of Essex, UK}

\date{}

\newtheorem{theorem}{Theorem}[section]
\newtheorem{corollary}[theorem]{Corollary}

\newtheorem{lemma}[theorem]{Lemma}

\newtheorem{observation}[theorem]{Observation}

\theoremstyle{definition}
\newtheorem*{comment*}{Comment}

\newcommand{\Gin}{G_{\text{in}}}
\newcommand{\Gout}{G_{\text{out}}}

\newcommand{\Ain}{A_{\text{in}}}
\newcommand{\Aout}{A_{\text{out}}}

\newcommand{\Sg}{P}

\newcommand{\Alg}{\text{ALG}}

\newcommand{\fav}{\text{top}}
\newcommand{\plu}{\text{plu}}

\newcommand{\SC}{\text{SC}}

\newcommand{\C}[1]{\text{SC}_#1}

\newcommand{\bsucc}{\boldsymbol{\succ}}

\newcounter{note}[section]

\begin{document}

\allowdisplaybreaks

\maketitle

\begin{abstract}
    We study the distortion of one-sided and two-sided matching problems on the line. In the one-sided case, $n$ agents need to be matched to $n$ items, and each agent's cost in a matching is their distance from the item they were matched to. We propose an algorithm that is provided only with ordinal information regarding the agents' preferences (each agent's ranking of the items from most- to least-preferred) and returns a matching aiming to minimize the social cost with respect to the agents' true (cardinal) costs. We prove that our algorithm simultaneously achieves the best-possible approximation of $3$ (known as distortion) with respect to a variety of social cost measures which include the utilitarian and egalitarian social cost. In the two-sided case, where the agents need be matched to $n$ other agents and both sides report their ordinal preferences over each other, we show that it is always possible to compute an optimal matching. In fact, we show that this optimal matching can be achieved using even less information, and we provide bounds regarding the sufficient number of queries.
\end{abstract}

\section{Introduction}
Matching problems are prevalent in everyday life and have played a central role in (computational) social choice for multiple decades, starting from the pioneering works of \citet{gale1962college}, \citet{shapley1974cores}, and \citet{hylland1979efficient}. These works consider settings in which a set of $n$ \emph{agents} need to be matched to a set of $n$ items (one-sided matching), or to another set of $n$ agents (two-sided matching), in some fair and efficient manner with respect to the agents' preferences. Some important applications of these problems include college admissions \citep{gale1962college}, national residency matching  \citep{nrmp,roth1984evolution,roth1999redesign}), school choice \citep{abdulkadirouglu2005boston,abdulkadirouglu2005new}, and organ exchange \citep{unos,roth2004kidney}. 

To accurately capture the, often complicated, preferences of an agent over a discrete set of outcomes, the common approach is to use a Von Neumann–Morgenstern utility function which assigns a numerical value to each outcome, indicating how strongly it is liked or disliked \citep{von1944theory}. Although the significance of these \emph{cardinal} preferences of the agents is recognized in the classic literature on matching (e.g., \citep{hylland1979efficient,zhou1990conjecture,bogomolnaia2001new}), most matching algorithms that have been proposed are \emph{ordinal}: They ask each agent to report only their ranking over the outcomes from most- to least-preferred.
This is motivated by the fact that it can be cognitively or computationally prohibitive for the agents to come up with these values, whereas simply ranking outcomes is an easier, often routine, task. However, given access only to this limited (ordinal) information regarding the agents' preferences, to what extent can one approximate social objectives which are functions of the cardinal values? 

This is precisely the question studied in the vibrant literature of \emph{distortion in social choice}, whose goal is to design ordinal algorithms with low \emph{distortion}, i.e., strong worst case approximation guarantees with respect to cardinal social cost objectives (see \cite{distortion-survey} for a recent survey). One of the most well-studied settings in this literature since its inception \citep{procaccia2006distortion} is that of \emph{metric distortion}, where the agents' cardinal preferences correspond to distances in some metric space \citep{merrill1999unified,enelow1984spatial}. The metric distortion literature has been very successful in a plethora of different settings, leading to (near-)optimal distortion bounds. Examples include single-winner voting \citep{anshelevich2018approximating,gkatzelis2020resolving,kempe2022veto}, multi-winner voting \citep{caragiannis2022multiwinner}, and probabilistic social choice \citep{feldman2016voting,anshelevich2017randomized,charikar2022randomized,charikar24breaking}. However, a notable exception that remains poorly-understood, despite its fundamental nature and wealth of applications, is the metric distortion of matching problems.

In the \emph{metric matching} setting, agents and items are embedded in a metric space, and an agent's cost in a matching is the distance between them and their match. The distortion of the metric (one-sided) matching problem was first studied by \citet{caragiannis2024augmentation}\footnote{The conference version of this paper was published in 2016.} who focused on the utilitarian social cost (the total cost) and, among other results, observed that the distortion of the well-known \emph{serial dictatorship} algorithm is $2^{n-1}$. For its randomized counterpart, the \emph{random serial dictatorship}, they showed that its distortion is between $n^{0.29}$ and $n$. Their work also implies a lower bound of $3$ on the distortion of any (possibly randomized) ordinal algorithm with respect to the utilitarian social cost. More recently, \citet{anari2023matching} achieved improvements on both fronts: They designed a deterministic algorithm with distortion $O(n^2)$ and complemented it with a lower bound of $\Omega(\log n)$ on the distortion of any ordinal algorithm, even randomized ones. Despite these improvements, there is still a very large gap between the lower and the upper bound, and closing this gap is one of the main open problems in this literature. 

In this paper, we resolve the metric distortion of the matching problem for the interesting case of the line metric, which has received a significant amount of attention, both because it models scenarios of interest and because it gives rise to worst-case instances for many metric distortion settings. The line metric, known as \emph{$1$-Euclidean}, provides an abstraction of political or ideological beliefs along different axes, encompassed into a spectrum, e.g., between ``left'' and ``right'', which is inherent in the pioneering works of \cite{hotelling1929stability} and \cite{downs1957economic}. The line metric can also be used to capture preferences over policy issues, such as the extent to which a government should be involved in the economy, ranging from full involvement to zero intervention \citep{stokes1963spatial}; see \citep{filos2024distortion} for a more detailed discussion.
It also turns out that in most of the fundamental metric distortion settings (e.g., single-winner and multi-winner voting), the best possible distortion on the line metric is either the same as or within a small constant from the best possible distortion on any metric. 
In the metric matching problem, however, this remains unknown: the lower bound of $\Omega(\log n)$ shown by \citet{anari2023matching} requires a highly involved tree metric, so it does not rule out the existence of an algorithm with \emph{constant} distortion on the line. Is a such a constant achievable? 

\subsection{Our Contribution}
Our main result provides a positive answer to this question by providing an ordinal algorithm that achieves a distortion of $3$ with respect to the well-known \emph{$k$-centrum cost} (for different values of $k \in \{1,\ldots,n\}$), which is equal to the sum of the $k$ largest agent costs~\citep{Tamir2001kcentrum}; see also \citep{Han2023multiple}. This objective naturally interpolates between the maximum (egalitarian) cost when $k=1$ and the utilitarian social cost when $k=n$. In fact, we show that this distortion is the best possible for each of these cost functions and even for randomized algorithms, which fully resolves this question.

\begin{inftheorem}
\emph{There is a deterministic ordinal algorithm for the 
one-sided matching problem on the line that simultaneously achieves distortion 3 with respect to the $k$-centrum cost for all $k\in \{1, \ldots, n\}$. This distortion is the best possible for any ordinal algorithm, including randomized ones.}
\end{inftheorem}

This result provides a clear, asymptotic separation between the distortion of one-sided matching on the line metric and its distortion in more general metric spaces. 

Building on our main result above, we then consider the distortion of the metric two-sided matching problem. Despite the ubiquitous nature of this problem, it is perhaps surprising that, to the best of our knowledge, its metric distortion has not been considered prior to our work. For this setting, we prove a seemingly surprising result, namely that, using only ordinal information, we can construct a matching that \emph{exactly} minimizes the $k$-centrum cost.

\begin{inftheorem}
\emph{There is a deterministic ordinal algorithm for the 
two-sided matching problem on the line that always returns an \emph{optimal} matching (i.e., achieves distortion $1$).}
\end{inftheorem}

From a technical standpoint, what enables us to improve from a distortion of $3$ in the one-sided matching case to achieving optimality in the two-sided matching setting is the additional information that we have from the preferences of the agents on the other side. We explore this further and show that an optimal matching can be computed using even less information about the ordinal preferences of the agents, which is acquired by appropriate {\em queries}.

\subsection{Further Related Work}
Prior to the work of \citet{anari2023matching}, \citet{anshelevich2021given} considered the metric one-sided matching problem as part of a class of more general facility assignment problems, and showed a bound of $3$ on the distortion of ordinal algorithms for the social cost and the maximum cost objectives, but under the restriction that the item locations in the metric space are known. Other works \citep{Anshelevich2019bipartitematching,Anshelevich2016truthful,anshelevich2016blind} considered related matching settings where the goal is to maximize the \emph{metric utilities} of the agents rather than minimizing their costs.

In the non-metric setting, where the agents are assumed to have normalized utilities over the items, the best possible distortion for one-sided matching has been identified for both deterministic~\citep{amanatidis2022matching} and randomized algorithms~\citep{filos2014RSD}. The distortion of one-sided matching (and generalizations of it) has also been studied when more information can be elicited via queries~ \citep{amanatidis2022matching,amanatidis2024dice,ma2021matching,latifian2024approval,ebadian2025bit}.
The query models in those works are inherently different from the one we consider here, as they elicit \emph{cardinal} information, on top of the ordinal preferences which are considered to be known. In contrast, our queries in the two-sided model elicit only \emph{ordinal} information about the agents' preferences on each side.  

\section{Preliminaries}
For any positive integer $\ell$, let $[\ell] = \{1,\ldots,\ell\}$. 
An instance of our problem consists of a set of $n$ {\em agents} $A=\{a_1,\dots,a_n\}$ and a set of $n$ {\em items} $G=\{g_1,\dots,g_n\}$. We assume that agents and items are located on distinct points on the same line metric\footnote{The assumption that agent and item locations are all distinct is without loss of generality: it suffices to assume that agent's ordinal preferences consistently tie-break over items with the same location.}, and let $d(a,g)$ denote the {\em distance} between any agent $a$ and item $g$. Let $\succ_a$ be the {\em ordinal preference} of agent $a$ over the set of items $G$, such that $g \succ_a g'$ implies $d(a,g) \leq d(a,g')$. Let $\bsucc := (\succ_a)_{a \in A}$ be the {\em ordinal profile} consisting of the ordinal preferences of all agents. Note that many different distance metrics may induce the same ordinal preferences for the agents. We write $d \rhd\!\!\bsucc$ to denote the event that the metric $d$ is consistent to the ordinal profile $\bsucc$. We also denote by $\fav(a)$ the {\em favorite item} of agent $a$, that is, $\fav(a) \succ_a g$ for every $g \in G \setminus\{\fav(a)\}$. Let $\plu(g)$ be the plurality score of an item $g\in G$, which is equal to the number of agents whose favorite item is $g$, i.e., $\plu(g)=\left|\left\{a\in A : \fav(a) = g\right\}\right|$.

Our goal is to choose a matching $M$ between the agents and the items. Given such a matching, we denote by $M(a)$ the item matched to agent $a$, and by $M(g)$ the agent matched to item $g$. We evaluate the quality a matching using a variety of social cost measures, captured by the $k$-centrum cost. Given some $k\in [n]$, the {\em k-centrum cost} of a matching $M$ is the sum of the $k$ largest distances between the agents and their matched items, among all agents: 
    \begin{equation*}
        \C{k}(M|d) = \sum_{q=1}^k {\max_{a\in A}}^q d(a,M(a)),
    \end{equation*}
where $\max^q$ returns the $q$-th largest value. Note that this captures the two most well-studied social cost functions, the egalitarian and the utilitarian social cost, as special cases. 
Specifically, $\C{n}(M|d)=\sum_a d(a, M(a))$ corresponds to the utilitarian social cost, which evaluates on the total cost over all agents, while $\C{1}(M|d)=\max_a d(a, M(a))$ corresponds to the egalitarian social cost, which focuses on the agent that suffers the largest cost. In general, smaller values of $k$ put more emphasis on fairness by focusing on the least happy agents. When the metric $d$ is clear from context, we will simplify our notation and write $\C{k}(M)$ instead of $\C{k}(M|d)$. 

A matching {\em algorithm} $\Alg$ takes as input an ordinal profile $\bsucc$ and computes a one-to-one matching $\Alg(\bsucc)$ of the items to the agents. We want to design matching algorithms with as small distortion as possible. 
The {\em distortion} of a matching algorithm $\Alg$ with respect to the $k$-centrum cost $\C{k}$ for some $k\in [n]$ is the worst-case ratio over all possible ordinal profiles and consistent metrics between the objective value of the matching returned by the algorithm and the minimum possible objective value over all possible matchings: 
\begin{align*}
    \sup_{\bsucc} \sup_{d: d \rhd \bsucc}~ \frac{\C{k}(\Alg(\bsucc)|d)}{\min_M \C{k}(M|d)}.
\end{align*}

\subsection{$M^*$: A well-structured optimal matching.}\label{sec:opt_m}
For each instance there may exist multiple matchings that minimize the $k$-centrum cost. Among these optimal matchings, we are particularly interested in one that greedily matches agents to items according to their true ordering on the line; the optimality of this matching is established in the following theorem. Throughout the paper, we denote this specific matching by $M^*$ and refer to it simply as the optimal matching. Also, for each agent $a$ we refer to $M^*(a)$ as $a$'s optimal item. Due to space constraints, the proof of the following theorem (as well as of others in the rest of the paper) are deferred to the supplementary material. 

\begin{theorem}\label{lem:opt_structure}
Given the true ordering of all agents and items on the line, greedily matching the leftmost agent to the leftmost item leads to an optimal matching $M^*$ with respect to $\C{k}$ for any $k\in [n]$.
\end{theorem}
\begin{proof}
Throughout this proof, with slight abuse of notation, we will use $a_i$ and $g_i$ to also denote the position of agent $a_i$ and item $g_i$ on the line for any $i \in [n]$. Without loss of generality, assume that $a_1 \leq  a_2 \leq  \ldots \leq a_n$ and $g_1 \leq g_2 \leq \ldots \leq g_n$. We claim that there exists an optimal matching $M^*$ such that $M^*(a_i) = g_i$ for any $i \in [n]$. Let $M$ be a different matching with strictly smaller $k$-centrum cost than $M^*$, and let $a_i$ be the leftmost agent that is not matched to its claimed optimal item, that is, $M(a_i) = g_\ell \neq g_i$. By the choice of $a_i$, we have that $M(a_j) = g_j$ for any $j < i$, which implies that $g_i \leq g_\ell$ and $M(a_j) = g_i$ for some agent $a_j$ with $j > i$, i.e., $a_i \leq a_j$.

We will now show that swapping the pairs $(a_i,g_\ell)$ and $(a_j,g_i)$ in $M$ leads to a new matching $M'$ where $a_i$ is also matched to its claimed optimal item $g_i$, and $\C{k}(M') \leq \C{k}(M)$ for any $k \in [n]$. Observe that, to show that the $k$-centrum cost does not increase after this swap for every $k \in [n]$, it suffices to show that the sum and the maximum of the distances of the involved agents from their matched items does not increase. Indeed, a non-increasing maximum implies that the $k$-centrum cost does not increase when one of the agents contributes to it in $M$. Similarly, a non-increasing sum implies that the $k$-centrum cost does not increase in case both agents contribute to it in $M$. Of course, if neither agent contributes to the $k$-centrum cost, then it remains the same after the swap. 

First observe that it cannot be the case that $g_\ell \succ_{a_i} g_i$ and $g_i \succ_{a_j} g_\ell$ as otherwise one of the two ordering constraints $g_i \leq g_\ell$ and $a_i \leq a_j$ would not be satisfied. 
Furthermore, if $g_i \succ_{a_i} g_\ell$ and $g_\ell \succ_{a_j} g_i$, then clearly the distance of both agents from their matched items would not increase after the swap. Consequently, it suffices to focus on the case where both agents prefer the same item out of the two, say $g_\ell$; the other case is symmetric. We now switch between the following three cases depending on the possible relative positions of the agents and the items. 
\begin{itemize}
    \item $g_i \leq a_i \leq a_j \leq g_\ell$.
    Since $d(a_i,g_\ell) = d(a_i,a_j) + d(a_j,g_\ell)$ and $d(a_j,g_i) = d(a_i,g_i) + d(a_i,a_j)$, the sum and the maximum of the distances clearly cannot increase after the swap. In particular, we have
    \begin{align*}
        d(a_i,g_\ell) + d(a_j,g_i) = d(a_i,g_i) + d(a_j,g_\ell) + 2 d(a_i,a_j)
    \end{align*}
    and 
    \begin{align*}
        \max\{d(a_i,g_\ell), d(a_j,g_i)\} = \max\{d(a_i,g_i), d(a_j,g_\ell)\} + d(a_i,a_j).
    \end{align*}
    
    \item $g_i \leq a_i \leq g_\ell \leq a_j$. The sum and the maximum of the distances clearly cannot increase after the swap since $d(a_j,g_i) = d(a_i,g_i) + d(a_i,g_\ell) + d(a_j,g_\ell)$. 

    \item $g_i \leq g_\ell \leq a_i \leq a_j$. In this case, we have
    $d(a_i,g_\ell) + d(a_i,a_j) = d(a_j,g_\ell)$, 
    $d(a_j,g_i) = d(a_i,g_i) + d(a_i,a_j)$ and $d(a_j,g_i) = d(g_i,g_\ell) + d(a_j,g_\ell)$. Using the first two equalities, we get that the sum remains the same after the swap: 
    \begin{align*}
        d(a_i,g_\ell) + d(a_j,g_i) 
        &= d(a_i,g_\ell) + d(a_i,g_i) + d(a_i,a_j) \\
        &= d(a_i,g_i) + d(a_j,g_\ell)
    \end{align*}    
    Using the last two equalities above, we get that the maximum cannot increase:
    \begin{equation*}
        \max\{d(a_i,g_\ell),d(a_j,g_i)\} = d(a_j,g_i) > \max\{d(a_i,g_i), d(a_j,g_\ell)\}.\qedhere
    \end{equation*}
\end{itemize}
\end{proof}

\subsection{Structural properties}
Given an ordinal profile $\bsucc$, we can observe some structure regarding the locations of the agents and the items on the line. 

\begin{lemma}\label{lem:range}
If the favorite item of two agents is the same, then all the agents between them on the line also have the same favorite item.
\end{lemma}
\begin{proof}
Consider any triplet of ordered agents $a_\ell, a_c, a_r$ such that $g$ is the favorite item of $a_\ell$ and $a_r$. 
If $g$ is to the left of $a_c$, then $g$ must be the favorite item of $a_c$ as well since $a_c$ is closer to $g$ than $a_r$. Similarly, if $g$ is to the right of $a_c$, then $a_c$ is closer to $g$ than $a_\ell$, and thus $g$ must also be $a_c$'s favorite item. 
\end{proof}

\begin{lemma}\label{lem:favorite_to_left}
If item $g_x$ is to the left of item $g_y$, then all agents whose favorite item is $g_x$ are to the left of all agents whose favorite item is $g_y$.
\end{lemma}
\begin{proof}
Let $A_x$ be the set of agents whose favorite item is $g_x$, and let $A_y$ be the set of agents whose favorite item is $g_y$. If an agent in $A_x$ is to the right of $g_y$, then this agent would be closer to $g_y$ than to $g_x$, thus contradicting that $g_x$ is this agent's favorite item. So, all agents in $A_x$ at to the left of $g_y$. Similarly, all agents in $A_y$ are to the right of $g_x$. The statement now follows by \cref{lem:range}, which implies that the agents in $A_x$ are all to the left of the agents of $A_y$.
\end{proof}

\section{One-sided case: Achieving the optimal distortion of $3$}
In this section we first prove that no algorithm can achieve a distortion better than 3 with respect to the $k$-centrum cost for any $k\in[n]$, even if it is randomized. Then, we present an optimal (deterministic) algorithm that simultaneously guarantees a distortion upper bound of $3$ with respect to all $k$-centrum costs. 

\subsection{Distortion lower bound for any algorithm}
Here we prove a lower bound on the distortion of matching algorithms, even randomized ones. 

\begin{theorem} \label{thm:lower}
For any $k \in [n]$, no deterministic algorithm can achieve a distortion better than $3$ with respect to the $k$-centrum cost $\C{k}$, and no randomized algorithm can achieve a distortion better than $3-2/n$.
\end{theorem}
\begin{proof}
Consider an instance with $n$ agents $\{a_1, \ldots, a_n\}$ and $n$ items $\{g_1, \ldots, g_n\}$. 
The ordinal profile is such that all agents have the same ranking over the items: $g_1 \succ \ldots \succ g_n$. 
For any $i \in [n]$, let $p_i$ be the probability that agent $a_i$ is matched to item $g_n$ according to an arbitrary randomized matching algorithm. Without loss of generality, suppose that $p_1 \leq 1/n$; note that if the algorithm is deterministic, then $p_1 = 0$.

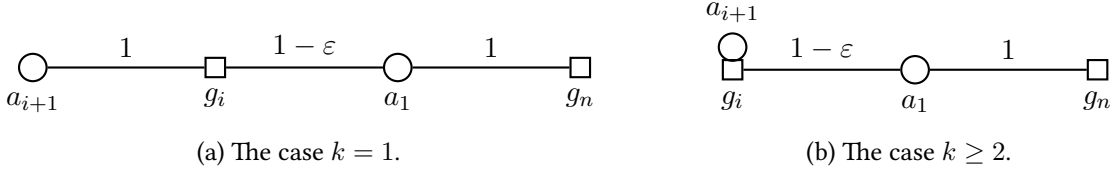
\begin{figure}[t]
\begin{subfigure}[t]{0.5\textwidth}
    \centering
    \begin{tikzpicture}[node distance={24mm}, thick, c/.style = {draw, circle, minimum size=2mm}, s/.style = {draw, rectangle, minimum size=2mm}, t/.style = {draw, regular polygon, regular polygon sides=3, minimum size=2mm]}, invisible/.style = {minimum size=0mm, inner sep=0mm, text width=0mm, outer sep=0mm, draw=none}]
    \node[c] [label=below:$a_{i+1}$] (1) {};
    \node[s] [right of=1, label=below:$g_i$] (2) {};
    \node[c] [right of=2, label=below:$a_1$] (3) {};
    \node[s] [right of=3, label=below:$g_n$] (4) {};
    \draw (1) -- node[midway, above] {$1$} (2);
    \draw (2) -- node[midway, above] {$1-\varepsilon$} (3);
    \draw (3) -- node[midway, above] {$1$} (4);
    \end{tikzpicture}
\caption{The case $k=1$.}
\label{fig:lower:k=1}
\end{subfigure}
\begin{subfigure}[t]{0.5\textwidth}
    \centering 
    \begin{tikzpicture}[node distance={24mm}, thick, c/.style = {draw, circle, minimum size=2mm}, s/.style = {draw, rectangle, minimum size=2mm}, t/.style = {draw, regular polygon, regular polygon sides=3, minimum size=2mm]}, invisible/.style = {minimum size=0mm, inner sep=0mm, text width=0mm, outer sep=0mm, draw=none}]
    \node[s] [label=below:$g_i$] (1) {};
    \node[c] [above of=1, node distance=3mm, label=above: $a_{i+1}$] (2) {};
    \node[c] [right of=1, label=below:$a_1$] (3) {};
    \node[s] [right of=3, label=below:$g_n$] (4) {};
    \draw (1) -- node[midway, above] {$1-\varepsilon$} (3);
    \draw (3) -- node[midway, above] {$1$} (4);
    \end{tikzpicture}
\caption{The case $k \geq 2$.}
\label{fig:lower:kgeq2}
\end{subfigure}
\caption{The metrics used in the proof of \cref{thm:lower} to give a lower bound of $3$ on the distortion of (randomized) algorithms for (a) $k=1$ and (b) $k\geq 2$. Circles correspond to agents, rectangles correspond to items, and $i \in [n-1]$. The weight of each edge is the distance between its endpoints.}
\label{fig:lower}
\end{figure}

Consider the metric spaces that are illustrated in \cref{fig:lower:k=1} for $k=1$ and in \cref{fig:lower:kgeq2} for $k\geq 2$, where $\varepsilon > 0$ is an infinitesimal. For any $k \in [n]$, the optimal matching consists of the pairs $(a_1,g_n)$ and $(a_{i+1},g_i)$ for $i \in [n-1]$, leading to a $k$-centrum cost of $d(a_1,g_n) = 1$. However, since the algorithm matches $a_1$ to $g_n$ with probability $p_1$, it matches agents in $\{a_1, \ldots, a_n\}$ to $g_n$ with the remaining probability $1-p_1$. Whenever an agent in $\{a_1, \ldots, a_n\}$ is matched to $g_n$, the $k$-centrum of this matching is $3-\varepsilon$, thus leading to an expected $k$-centrum cost of
\begin{align*}
    p_1 \cdot 1 + (1-p_1) \cdot (3-\varepsilon) = 3-\varepsilon -(2-\varepsilon) p_1.
\end{align*}
Consequently, the distortion is at least $3-\varepsilon$ for deterministic algorithms (since $p_1=0)$ and $3-2/n - \varepsilon$ for randomized algorithms, for any $\varepsilon > 0$.  
\end{proof}

\begin{remark}
Using an instance similar to the one depicted in \cref{fig:lower:kgeq2} and strategyproofness arguments (rather than orginality), \citet{caragiannis2024augmentation} proved the same lower bound of $3$ for the utilitarian social cost ($k=n$). Our lower bound is more general as it applies to all values of $k$. 
\end{remark}

\subsection{Distortion-optimal matching algorithm} \label{subsec:algorithm}
We now provide a deterministic algorithm that simultaneously achieves the optimal distortion of $3$ with respect to $\C{k}$ for any $k\in [n]$. Our algorithm consists of the following three key steps:
The first one is to appropriately partition the items into two sets $\Gin$ and $\Gout$ using the ordinal preferences of the agents.
The second step determines the true ordering (up to reversal) of the items in $\Gin$, and it also uses this ordering to define an ordering over the set of agents. 
The third step matches the items in $\Gin$ to agents using the aforementioned orderings, and then arbitrarily matches the remaining items to the remaining agents. 
The rest of this section provides more details regarding the implementation of these steps. A description of the algorithm using pseudocode is given as Algorithm~\ref{algname:order_match}. 

\smallskip
\noindent 
{\bf Step 1: Partitioning the items into $\Gin$ and $\Gout$.} \ \\
Let $G_+ := \{g \in G: \plu(g) \geq 1 \}$ be the subset of items in $G$ that have a positive plurality score. By considering each agent's ranking of the items in $G_+$, at most two items in $G_+$ can be ranked last by the agents (these would be the two extreme items in $G_+$, i.e., the leftmost and the rightmost one). We use $g_\ell$ and $g_r$ to denote these items, and without loss of generality\footnote{Note that both the behavior of our algorithm and the analysis is invariant to the reversal of agent and item locations; so, we only fix $g_\ell$ to be the leftmost item to make the exposition more intuitive.}, we assume that the former is the leftmost item in $G_+$ and the latter is the rightmost one. If every agent ranks the same item of $G_+$ last, that means $|G_+| = 1$ and $g_\ell = g_r$. 

Since $g_\ell,g_r\in G_+$, both of these items have at least one agent who ranks them at the top. Given two agents $a_i$ and $a_j$ such that $\fav(a_i) = g_\ell$ and $\fav(a_j) = g_r$, 
we define the following subset of items:
\begin{equation*}
\Gin(a_i, a_j) := \{g \in G: g \succ_{a_i} g_r \text{ or } g \succ_{a_j} g_\ell\}. 
\end{equation*}
We define $(a_\ell, a_r) \in \arg\max_{a_i,a_j}|\Gin(a_i,a_j)|$ to be two agents that maximize the size of $\Gin(a_i, a_j)$, we henceforth partition the set of agents into the following two sets:
\begin{equation}\label{def:GinGout}
    \Gin := \Gin(a_\ell, a_r) ~~~~\text{and}~~~~ \Gout := G \setminus \Gin. 
\end{equation}

In words, $\Gin$ consists of the items that either agent $a_\ell$ prefers over $g_r$ or agent $a_r$ prefers over $g_\ell$, and $\Gout$ consists of all the remaining items. Note that defining $(a_\ell, a_r) \in \arg\max_{a_i,a_j}|\Gin(a_i,a_j)|$ is vital to achieving a distortion of 3. If $a_\ell$ and $a_r$ were to be chosen arbitrarily among the agents whose favorite items are $g_\ell$ and $g_r$, a distortion of 3 would not be achieved with respect to $\C{k}$ for any $k \in [n]$ (see \cref{obs:tiebreak} in the appendix).

Note that if $g_\ell=g_r$, i.e., there is only one item in $G_+$, then $\Gin=\varnothing$ since no other item is preferred over that item.

\begin{lemma}\label{lem:G+SubsetGin}
If $\Gin\neq\varnothing$, then $G_+\subseteq \Gin$.
\end{lemma}
\begin{proof}
Note that all items in $G_+$ are between $g_\ell$ and $g_r$, $g_r$ is to the right of $g_\ell$, and $\fav(a_\ell)=g_\ell$, so all items in $G_+$ except $g_\ell$ must be to the right of $a_\ell$ because if any were to the left of $a_\ell$ they would either be to the left of $g_\ell$ or between $a_\ell$ and $g_\ell$ which are contradictions, therefore all items in $G_+$ except $g_\ell$ must be between $a_\ell$ and $g_r$, so $a_\ell$ prefers them to $g_r$; $g_\ell$ and $g_r$ must be in $\Gin$ because they are $a_\ell$'s and $a_r$'s favorite item respectively, and $a_\ell$ and $a_r$ prefer their favorite item to $g_r$ and $g_\ell$ respectively.   
\end{proof}

\paragraph{Step 2(a): Ordering $\pi_g$ of the items in $\Gin$.} \ \\ 
Although we cannot fully uncover the true ordering of all items in $G$ given the ordinal preferences of the agents, we can extract the true ordering $\pi_g$ of the items in $\Gin\subseteq G$ using the algorithm of \cite{EF14} which do so for a superset of $\Gin$. 

\begin{lemma}\label{lem:order}
    The true ordering $\pi_g$ of the items in $\Gin$ can be computed.
\end{lemma}
\begin{proof}
   In the proof of their Theorem 1, \citet{EF14} define a set $C^+$ consisting of the items colored by their algorithm.
   Formally, $C^+$ is defined as the set of items on which the leftmost and the rightmost agents do not agree. This means that, for a pair of items $g$ and $g'$, if the leftmost agent prefers $g$ to $g'$ and the rightmost agent prefers $g'$ to $g$, then both $g$ and $g'$ are included in $C^+$. Their algorithm computes a complete ordering of this set (up to reversal). We argue that $\Gin \subseteq C^+$ by showing that for each $g \in \Gin$ there exists $g' \in G$ that the leftmost and the rightmost agents disagree on $g$ and $g'$. Hence, the algorithm of \citet{EF14} outputs a complete ordering of the items in $\Gin$. 
   
   Let $a_\ell^*$ be the leftmost agent and $a_r^*$ be the rightmost agent. If $a_\ell^*$ is to the right of $g_r$ or $a_r^*$ is to the left of $g_\ell$, then $g_\ell=g_r$ and $\Gin = \varnothing \subseteq C^+$. Otherwise, observe that $\fav(a_\ell)=\fav(a_\ell^*)=g_\ell$ and $\fav(a_r)=\fav(a_r^*)=g_r$. Consider $g \in \Gin$ and assume that $g \succ_{a_\ell} g_r$. Because $a_\ell$ is to the left of $g_r$ and prefers $g$ to $g_r$, $g$ must be to the left of $g_r$. If $g$ is to the right of $a_\ell^*$ then $a_\ell^*$ must prefer $g$ to $g_r$ because $g$ is between $a_\ell^*$ and $g_r$. If $g$ is to the left of $a_\ell^*$ then $d(a_\ell^*,g) \leq d(a_\ell,g)$ because $a_\ell$ is to the right of $a_\ell^*$. $a_\ell$ preferring $g$ to $g_r$ implies $d(a_\ell,g) \leq d(a_\ell,g_r)$, and $a_\ell$ being between $a_\ell^*$ and $g_r$ implies $d(a_\ell,g_r) \leq d(a_\ell^*,g_r)$. Thus, $d(a_\ell^*,g) \leq d(a_\ell^*,g_r)$, which implies $a_\ell^*$ prefers $g$ to $g_r$. Since $g_r$ is $a_r^*$'s favorite item, $a_r^*$ prefers $g_r$ to $g$. Hence $a_\ell^*$ and $a_r^*$ disagree which gives us the desired condition.

Now, in Proposition 1 of their paper, \citet{EF14} claim that the order that can be implied from single-peaked preferences on the voters is unique up to reversal. In addition, in the proof of Theorem 1, they show that any mapping of the voters to $\mathbb R$ that is consistent with their ordering, leads to the same ordering on $C^+$, which means the ordering of $C^+$ computed by their algorithm is unique (up to reversal) and gives us the true relative ordering of $\Gin$.
\end{proof}

\smallskip
\noindent 
{\bf Step 2(b): Ordering $\pi_a$ of the agents in $A$.} \ \\
Given the ordering $\pi_g$ of the items in $\Gin$, we also define an ordering of all the agents in $A$ based on what their top item is. Note that every agent's top item is in $G_+$ which, by \cref{lem:G+SubsetGin} is a subset of $\Gin$, so by \Cref{lem:favorite_to_left}, $\pi_g$ implies a partial order over $A$. First, we let $\tilde \pi_a$ be the partial ordering of the agents in $A$ such that $\tilde \pi_a(a_i)  = \tilde \pi_a(a_j)$ if $\fav(a_i) = \fav(a_j)$ and $\tilde \pi_a(a_i)  > \tilde \pi_a(a_j) $ if $ \pi_g(\fav(a_i)) > \pi_g(\fav(a_j))$. Then, we define a total order $\pi_a$ by breaking the ties of the partial order in some arbitrary way.

\smallskip
\noindent 
{\bf Step 3: Matching agents to items using $\pi_a$ and $\pi_g$.} \ \\
Our algorithm then determines how to match the items to agents starting with the items in $\Gin$. It considers them based on the ordering $\pi_g$ and assigns them to the agents using the ordering of $\pi_a$. Once all the items in $\Gin$ have been assigned, our algorithm arbitrarily matches the items in $\Gout$ to the agents that remain unmatched.

\begin{algorithm}[t]
    \caption{OrderMatch}
    \crefformat{mech1}{#2OrderMatch#3}

    \label{algname:order_match}
    \label{algnum:order_match}
    \label[mech1]{algname:order_match}
    \label[mechanism]{algnum:order_match}

    \KwIn{Sets $A$ and $G$ and ordinal profile $\bsucc$
    }
    \KwOut{A matching $M$ between $A$ and $G$}
    
    Identify the extreme items $g_\ell$ and $g_r$ using $\bsucc$.

    Partition $G$ into $\Gin$ and $\Gout$ according to \cref{def:GinGout}.
    Compute ordering $\pi_g$ over $\Gin$ and $\pi_a$ over $A$.

    Match the $i$th item in $\pi_g$ to the $i$th agent in $\pi_a$.

    Arbitrarily match items in $\Gout$ to unmatched agents.
\end{algorithm}

\subsection{Analysis of \cref{algname:order_match}}\label{sec:analysis}
We now prove the main result of this section. 
\begin{theorem}\label{thm:distortion(UB}
Algorithm~\cref{algname:order_match} achieves a distortion of 3 with respect to the $k$-centrum cost $\SC_k$ for any $k\in [n]$.
\end{theorem}

While our algorithm only has access to the ordinal preferences of the agents through $\bsucc$, our analysis in this section uses the exact locations of the agents and items (e.g., to further partition the agents and items using this information, as well as to define a graph induced by each matching).
For example, we further partition the items in $\Gout$ into $\Gout^\ell$ and $\Gout^r$ such that 
$\Gout^\ell$ consists of the items in $\Gout$ that are to the left of the items in $\Gin$; similarly, $\Gout^r$ consists of the items in $\Gout$ that are to the right of those in $\Gin$. 
Also, let $\Ain=\{a\in A : M^*(a)\in \Gin\}$ be the set of agents whose optimal items (according to the optimal matching $M^*$) are in $\Gin$ and $\Aout=A\setminus \Ain$ be the set of all the remaining agents, i.e., the agents whose optimal items are in $\Gout$. $\Aout$ is further partitioned into $\Aout^\ell$ and $\Aout^r$, which are the sets of agents whose optimal items are in $\Gout^\ell$ and $\Gout^r$, respectively.

\smallskip
\noindent 
{\bf Permutation graph $\Sg_M$ induced by matching $M$.} \ \\
Given an instance with optimal matching $M^*$, we define the \emph{permutation-graph} of any matching $M$ as $\Sg_M=(A,E)$, where each vertex corresponds to an agent in $A$, and each directed edge $(a_i,a_j)$ exists in $E$ if and only if $M(a_i)=M^*(a_j)$ (i.e. there exists an edge from agent $a_i$ to the agent $a_j$ if and only if $M$ assigns $a_j$'s optimal item to $a_i$; note that it may be the case that $a_i=a_j$). In other words, this graph captures the permutation of items from $M^*$ to $M$. Note that for any given instance (i.e., for a fixed $M^*$) every permutation graph $\Sg_M$ corresponds to a unique matching $M$. Also, note that every vertex of a permutation-graph has in-degree and out-degree equal to $1$ and $\Sg_M$ is a collection of disjoint cycles. 

To categorize some of the edges of this graph, we further refine the set of agents $\Ain$ as follows: Let $A_x \subseteq \Ain$ be the subset of agents in $\Ain$ whose favorite item is the $x$th-leftmost item in $\Gin$. 
Due to \cref{lem:favorite_to_left}, every agent in $A_x$ is to the left of every agent in $A_y$ for any $x < y$. We refer to an edge $(a_i,a_j) \in E$ as
\begin{itemize}
\item \textit{forward} if $a_i \in A_x$ and $a_j \in A_y$ for some $x<y$.
\item \textit{backward} if $a_i \in A_x$ and $a_j \in A_y$ for some $x>y$.
\item \textit{internal} if $a_i \in A_x$ and $a_j \in A_x$. 
\item \textit{inward} if $a_i \in \Aout$ and $a_j \in \Ain$. 
\end{itemize}
See the example in \Cref{fig:S-graph}. 

For any fixed instance, each graph $\Sg_M$ corresponds to a distinct matching $M$, so we can also express the $k$-centrum of a graph $\Sg_M$ and its corresponding matching $M$ as: 
$$\SC_k(\Sg_M) = \SC_k(M) = \sum_{q=1}^k \max\limits_{(a_i,a_j) \in E}^{\qquad q} d(a_i,M^*(a_j))$$ 

\begin{figure*}
    \centering
    \scalebox{0.85}{
    \begin{tikzpicture}[node distance={10mm}, thick, c/.style = {draw, circle, minimum size=2mm}, s/.style = {draw, rectangle, minimum size=2mm}, t/.style = {draw, regular polygon, regular polygon sides=3, minimum size=2mm]}, invisible/.style = {minimum size=0mm, inner sep=0mm, text width=0mm, outer sep=0mm, draw=none}]
        \node[invisible] (left) {};
        \node[c] [right of=left] (aol1) {};
        \node[c] [right of=aol1] (aol2) {};
        \node[c] [right of=aol2] (a1_1) {};
        \node[c] [right of=a1_1] (a1_2) {};
        \node[c] [right of=a1_2] (a1_3) {};
        \node[c] [right of=a1_3] (a1_4) {};
        \node[c] [right of=a1_4] (a1_5) {};
        \node[c] [right of=a1_5] (a2_1) {};
        \node[c] [right of=a2_1] (a2_2) {};
        \node[c] [right of=a2_2] (a2_3) {};
        \node[c] [right of=a2_3] (aor1) {};
        \node[c] [right of=aor1] (aor2) {};
        \node[invisible] [right of=aol2, node distance=3mm] (raol2) {};
        \node[invisible] [left of=a1_1, node distance=3mm] (la1_1) {};
        \node[invisible] [right of=a1_5, node distance=3mm] (ra1_5) {};
        \node[invisible] [left of=a2_1, node distance=3mm] (la2_1) {};
        \node[invisible] [right of=a2_3, node distance=3mm] (ra2_3) {};
        \node[invisible] [left of=aor1, node distance=3mm] (laor1) {};
        \node[invisible] [right of=aor2] (right) {};
        \node[invisible] [below of=left, node distance=5mm] (b_left) {};
        \node[invisible] [below of=raol2, node distance=5mm] (b_raol2) {};
        \node[invisible] [below of=la1_1, node distance=5mm] (b_la1_1) {};
        \node[invisible] [below of=ra1_5, node distance=5mm] (b_ra1_5) {};
        \node[invisible] [below of=la2_1, node distance=5mm] (b_la2_1) {};
        \node[invisible] [below of=ra2_3, node distance=5mm] (b_ra2_3) {};
        \node[invisible] [below of=right, node distance=5mm] (b_right) {};
        \node[invisible] [below of=laor1, node distance=5mm] (b_laor1) {};
        \draw[dashed, RedOrange] (raol2) -- (b_raol2);
        \draw[->, dashed, RedOrange] (b_raol2) -- node[midway, below, black] {$\Aout^\ell$} (b_left);
        \draw[dashed, RedOrange] (la1_1) -- (b_la1_1);
        \draw[dashed, RedOrange] (b_la1_1) -- node[midway, below, black] {$A_1$} (b_ra1_5);
        \draw[dashed, RedOrange] (b_ra1_5) -- (ra1_5);
        \draw[dashed, RedOrange] (la2_1) -- (b_la2_1);
        \draw[dashed, RedOrange] (b_la2_1) -- node[midway, below, black] {$A_2$} (b_ra2_3);
        \draw[dashed, RedOrange] (b_ra2_3) -- (ra2_3);
        \draw[dashed, RedOrange] (laor1) -- (b_laor1);
        \draw[->, dashed, RedOrange] (b_laor1) -- node[midway, below, black] {$\Aout^r$} (b_right);
        \draw[NavyBlue, ->] (aol2) to [in=90, out=90] node[midway, above] {Inward} (a1_3);
        \draw[YellowOrange, ->] (a1_2) to [in=90, out=90] node[midway, above] {Internal} (a1_1);
        \draw[RawSienna, ->] (a1_5) to [in=90, out=90] node[midway, above] {Forward} (a2_1);
        \draw[OliveGreen, ->] (a2_2) to [in=90, out=90] node[midway, above] {Backward} (a1_4);
        \draw[NavyBlue, ->] (aor2) to [in=90, out=90] node[midway, above] {Inward} (a2_3);
    \end{tikzpicture}
    }
    \caption{A snapshot of a graph $\Sg_M$ induced by some matching $M$, with its nodes appearing in the order of the corresponding agents' true locations on the line. The dashed lines below the nodes exhibit how these agents are partitioned into $\Aout^\ell$, $A_1$, $A_2$, and $\Aout^r$. A subset of the graph's edges appear above the nodes, labeled by their type (forward, backward, internal, and inward).}
    \label{fig:S-graph}
\end{figure*}
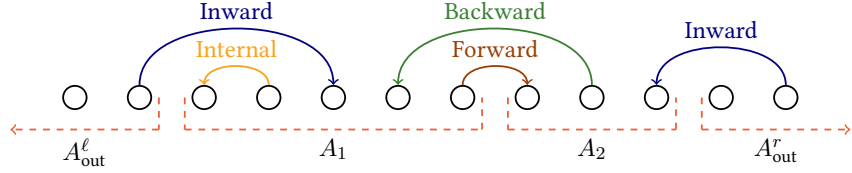

We now prove some properties about the permutation-graph $\Sg_M$ induced by our algorithm, which will be useful in proving the distortion bound. 

\begin{lemma}\label{lem:back-edges}
If $M$ is the matching computed by \cref{algname:order_match}, then the induced graph $\Sg_M$ does not include backward edges.
\end{lemma}

\begin{proof}
    If $\Gin$ is empty (i.e., $g_\ell = g_r$), then $\Sg_M$ cannot include any backward edges. If $\Gin$ is non-empty, then \cref{lem:order} shows that $\pi_g$ correctly orders the items in $\Gin$ from left to right. In the ordering $\pi_a$, agents are ordered earlier than other agents if their favorite item is to the left of the favorite item of those other agents according to $\pi_g$. For any $x < y$, \cref{lem:favorite_to_left} shows that all agents in $A_x$ are to the left of all agents in $A_y$. Thus, the optimal item of every agent in $A_x$ is to the left of the optimal item of every agent in $A_y$ due to \cref{lem:opt_structure}. Consequently, since all agents in $A_x$ are matched to items before the agents in $A_y$, no agent in $A_y$ can be matched to the optimal item of an agent in $A_x$, and thus there are no backward edges in $\Sg_M$.
\end{proof}

\paragraph{Reduction to graphs without forward edges.} \ \\
If $\Sg_M$ is the permutation-graph induced by the matching $M$ computed by the \cref{algname:order_match} algorithm, \cref{lem:back-edges} shows that this graph will not contain any backward edges. We now prove that, even though this graph can, in general, contain forward edges, we can without loss of generality assume that it does not. Specifically, we provide a reduction (\cref{algname:remove_forward_edges}) that takes as input any permutation-graph $\Sg_M$ without backward edges and transforms it into a graph $\Sg_{M'}$ that contains neither forward or backward edges, while ensuring that $\C{k}(\Sg_M) \leq \C{k}(\Sg_{M'})$ for any $k \in [n]$. This reduction iteratively removes inward or forward edges from the graph and replaces them with other inward or internal edges, until no forward edges remain.

\begin{theorem}\label{lem:forward-edges}
Let $M$ be the matching computed by \cref{algname:order_match} when given an ordinal profile $\bsucc$. 
If its permutation-graph $\Sg_M$ contains a forward edge, then there exists another matching $M'$ whose permutation-graph $\Sg_{M'}$ does not contain any forward or backward edges, and has weakly larger $k$-centrum cost, i.e., 
$\C{k}(M')\geq \C{k}(M)$ for any $k \in [n]$. 
\end{theorem}

\begin{algorithm}
    \SetAlgorithmName{Algorithm}{}{}
    \caption{RemoveForwardEdges reduction}
    
    \label{algname:remove_forward_edges}
    \label[alg2]{algname:remove_forward_edges}
    \label{algnum:remove_forward_edges}
    
    \KwIn{Graph $\Sg_M=(A,E)$ with no backward edges.}
    \KwOut{Graph $\Sg_{M'}$.}
    \While{there exists a forward edge in $E$}{
        Let $(a_3,a_4)\in E$ be a forward edge such that $\fav(a_3)$ is weakly to the right of $\fav(a_i)$ for any other forward edge $(a_i, a_j)\in E$.

        Find a forward or inward edge $(a_1,a_2) \in E$ such that $\fav(a_2)=\fav(a_3)$ \label{line:conditions}.
            
        $E \gets E \setminus \{(a_1,a_2), (a_3,a_4)\}$.
        
        $E \gets E \cup \{(a_1,a_4), (a_3,a_2)\}$.
    }
    \Return $\Sg_{M'} \gets (A, E)$.
\end{algorithm}

\begin{lemma}\label{lem:rfe_agents}
    While a forward edge $(a_3, a_4)$ exists in $E$, there must also exist a forward or inward edge $(a_1, a_2)\in E$ with $\fav(a_2)=\fav(a_3)$, i.e., an edge satisfying the  \cref{line:conditions} conditions of \cref{algname:remove_forward_edges}.
\end{lemma}

\begin{proof}
    Since $(a_3, a_4)$ is a forward edge, this implies that both $a_3$ and $a_4$ are in $\Ain$ and their favorite items are in $\Gin$, with $\fav(a_3)$ being to the left of $\fav(a_4)$. The existence of the forward edge $(a_3, a_4)$ implies that the item $a_3$ was matched to by \cref{algname:order_match} resulted in a forward edge in $E$. Either $a_3$ was matched to the optimal item of agent $a_4$ by \cref{algname:order_match}, or some previous iteration of \cref{algname:remove_forward_edges} added edge $(a_3,a_4)$ to $E$. For the latter to be the case, in some previous iteration of \cref{algname:remove_forward_edges} that added edge $(a_3,a_4)$ to $E$, edges $(a_1',a_2')$ and $(a_3',a_4')$ were replaced by edges $(a_1',a_4')$ and $(a_3',a_2')$. $(a_3',a_2')$ is guaranteed to be an internal edge, so $a_3$ must have been agent $a_1'$. Because $a_3 \in \Ain$, $(a_1',a_2')=(a_3,a_2')$ must have been a forward edge (rather than an inward edge). $a_3$ must have been matched to $a_2'$ by \cref{algname:order_match} because the only other way for $(a_3,a_2')$ to have been in $E$ would be if it were added to $E$ in an even earlier iteration of \cref{algname:remove_forward_edges} which would be in contradiction with \cref{algname:remove_forward_edges} iterating over forward edges right to left because it would require some other forward edge $(a'',a_2')$ to have been removed where $\fav(a'')$ is to the left of $\fav(a_2')=\fav(a_3')$ which \cref{lem:favorite_to_left} implies $a''$ is to the left of $a_3'$.

    Let $a_3 \in A_x$. As argued above, the item $a_3$ was originally matched to by \cref{algname:order_match} resulted in a forward edge. Let the agent whose optimal item $a_3$ was matched to be in $A_y$. \cref{lem:favorite_to_left} shows that every agent in $A_x$ is to the left of every agent in $A_y$ (which implies every agent in $A_x$'s optimal item is to the left of every agent in $A_y$'s optimal item). Thus, for $a_3$ to have been matched with an agent in $A_y$'s optimal item, all agents in $A_x$'s optimal items must have already been matched by the time $a_3$ was matched in \cref{algname:order_match}, which implies some agent not in $A_x$ was matched to an agent in $A_x$'s optimal item. \cref{lem:back-edges} shows that $\Sg_M$ induced by the matching $M$ returned from \cref{algname:order_match} contains no backward edges. Therefore, the agent who was matched to an agent in $A_x$'s optimal item must be in some $A_w$, such that $w<x$, or $\Aout$. Because \cref{algname:remove_forward_edges} iterates over forward edges from right to left, the edge induced from this matching must still be in $E$, let this edge be $(a_1,a_2)$. $a_1 \in A_w \cup \Aout$ and $a_2 \in A_x$, so both properties, that the edge is forward or inward, and $\fav(a_2)=\fav(a_3)$, are held.
\end{proof}

\begin{lemma}\label{lem:rfe_terminate}
    The \cref{algname:remove_forward_edges} reduction always terminates and returns a permutation-graph.
\end{lemma}
\begin{proof}
    Note that after each iteration of the while-loop, a forward edge is removed and no forward edge is added. Therefore, since the total number of edges is $n$ (one outgoing edge per agent), the algorithm will have to terminate after at most $n$ iterations. Also, after each iteration of the while loop in the reduction, the in-degree and the out-degree of every node remains equal to $1$, so the resulting graph will also be a permutation-graph for some new matching $M'$.
\end{proof}

\begin{lemma}\label{lem:rfe_backward}   
\cref{algname:remove_forward_edges} does not introduce any backward edges.
\end{lemma}

\begin{proof}
     In each iteration of its while loop, \cref{algname:remove_forward_edges} replaces edges $(a_1,a_2)$ and $(a_3,a_4)$ from $E$ with edges $(a_1,a_4)$ and $(a_3,a_2)$. The properties of these edges are that $(a_1,a_2)$ is either a forward or inward edge, $(a_3,a_4)$ is a forward edge, and $\fav(a_2)=\fav(a_3)$. Because $(a_3,a_4)$ is a forward edge and $\fav(a_2)=\fav(a_3)$, agents $a_2,a_3$, and $a_4$ must be in $\Ain$. Therefore, $(a_3,a_2)$ is an internal edge. If $(a_1,a_2)$ is a forward edge then $(a_1,a_4)$ must also be a forward edge because $a_1 \in \Ain$, and $\fav(a_1)$ is to the left of $\fav(a_4)$ ($\fav(a_1)$ is to the left of $\fav(a_2)$, $\fav(a_3)$ it to the left of $\fav(a_4)$, and $\fav(a_2) = \fav(a_3)$). If $(a_1,a_2)$ is an inward edge then $a_1 \in \Aout$, and because $a_4 \in \Ain$, $(a_1,a_4)$ is an inward edge. Neither of the edges, $(a_1,a_4)$ or $(a_3,a_2)$, that \cref{algname:remove_forward_edges} adds to $E$ are backward edges, so \cref{algname:remove_forward_edges} does not introduce any backwards edges.
\end{proof}

Our next step is to also prove that the $k$-centrum cost weakly increases during the execution of the \cref{algname:remove_forward_edges} reduction. To prove that this is the case, we first prove a helper lemma regarding any inward edges from $\Aout^r$ that the graph may contain.

\begin{lemma}\label{lem:inward-right}
    If $(a_i,a_j) \in E$ is an inward edge s.t.\ $a_i \in \Aout^r$, then $\fav(a_j)$ is weakly to the right of $\fav(a)$ for any agent $a\in \Ain$.
\end{lemma}
\begin{proof}
    First, note that the true location of $a_i$ must be to the right of all agents in $\Ain$ because $M^*(a_i) \in \Gout^r$ is to the right of all items in $\Gin$, and \cref{lem:opt_structure} shows that if an agent's optimal item is to the right of another agent's optimal item then that agent is to the right of the other agent. Therefore, \cref{lem:favorite_to_left} implies that $\fav(a_i)$ cannot be to the left of the favorite item of any agent in $\Ain$. If $\fav(a_i)$ is to the right of all agents in $\Ain$'s favorite items, then at least $|\Ain|$ agents were ordered before $a_i$ in $\pi_a$, and thus all items in $\Gin$ would have been matched before $a_i$ could have been matched. Therefore, for $a_i$ to have been matched to an item in $\Gin$, $\fav(a_i)$ must be the same item as the rightmost agent in $\Ain$'s favorite item. When creating the ordering $\pi_a$, \cref{algname:order_match} arbitrarily orders agents who share the same favorite item, so an alternate ordering, $\pi_a'$, can be created by swapping the ranking of $a_i$ and any agent in $\Ain$ who shares the same favorite item, and $\pi_a'$ would be a possible ordering created by \cref{algname:order_match}.  \cref{lem:back-edges} implies that no agent, $a_x$, in $\Ain$ can be matched to another agent, $a_y$, in $\Ain$'s optimal item, if $\fav(a_x)$ is to the right of $\fav(a_y)$. Thus, if $\pi_g$ were to be matched to $\pi_a'$, the agent swapped with $a_i$ would not be matched to the optimal item of any agent in $\Ain$ whose favorite item is to the left of their favorite item. Therefore, when matching $\pi_g$ to $\pi_a$, $a_i$ will not match to the optimal item of any agent in $\Ain$ whose favorite item is to the left of their favorite item. Thus, for $a_i$ to be matched with $M^*(a_j)$, $\fav(a_j)$ must be $\fav(a_i)$, which has been shown to be weakly to the right of $\fav(a)$ for all $a \in \Ain$.
\end{proof}

\begin{lemma}\label{lem:rfe_cost}
    After each iteration of the while loop in \cref{algname:remove_forward_edges} the $k$-centrum cost of the new matching weakly increases for any $k \in [n]$.
\end{lemma}

\begin{proof}
    Let $\Sg_M$ be the graph at the beginning of some iteration of the while loop of \cref{algname:remove_forward_edges} and let $\Sg_{M'}$ be the updated graph after having its edges changed. We will show that $\C{k}(\Sg_M) \leq \C{k}(\Sg_{M'})$ for any $k\in [n]$. 
    
    To do this, we will first show that the maximum cost between the two edges, $(a_1, a_4)$ and $(a_3, a_2)$, added to $\Sg_{M'}$ is weakly greater than the maximum between the two removed edges,  $(a_1, a_2)$ and $(a_3, a_4)$, i.e.,

    \begin{equation}\label{ineq:max}
        \max(d(a_1,M^*(a_2)),d(a_3,M^*(a_4))) \leq \max(d(a_1,M^*(a_4)),d(a_3,M^*(a_2))).
    \end{equation}
    We will then also show the same for the sum of their costs, i.e.
    \begin{equation}\label{ineq:sum}
        d(a_1,M^*(a_2))+d(a_3,M^*(a_4)) \leq d(a_1,M^*(a_4))+d(a_3,M^*(a_2)).
    \end{equation}
    We will show these by arguing that agents $a_1$, $a_2$, $a_3$, and $a_4$ are in a specific order on the line (and thus so are their optimal items given the optimal structure of \cref{lem:opt_structure}). Then, the ordering shown will imply which matching of $a_1$ and $a_3$ to $M^*(a_2)$ and $M^*(a_4)$ is optimal for both the sum and max (again using \cref{lem:opt_structure}).

    \emph{Agent $a_1$ must be to the left of both $a_2$ and $a_3$}. This is true because $(a_1,a_2)$ is either a forward edge or an inward edge. If it is a forward edge, this implies $\fav(a_1)$ is to the left of $\fav(a_2)$. Combined with the fact that $\fav(a_2)=\fav(a_3)$, and with \cref{lem:favorite_to_left} (which shows that if an agent's favorite item is to the left of another agent's favorite item then that agent is also to the left of the other agent) we conclude that $a_1$ must be to the left of both $a_2$ and $a_3$. If, on the other hand, $(a_1,a_2)$ is an inward edge, given that $(a_3,a_4)$ is a forward edge and $\fav(a_2)=\fav(a_3)$, \cref{lem:inward-right} implies $a_1 \in \Aout^\ell$ (rather than $\Aout^r$) and thus $M^*(a_1) \in \Aout^\ell$. Given that $\{a_2,a_3\} \in \Ain$, which implies $\{M^*(a_2),M^*(a_3)\} \in \Gin$, all items in $\Gout^\ell$ are to the left of $\Gin$ by definition, and $M^*$ is defined using the optimal structure from \cref{lem:opt_structure} which implies if an agent's optimal item is to the left of another agent's optimal item then that agent is also to the left of the other agent, then $a_1$ must be to the left of $a_2$ and $a_3$.

    \emph{Agent $a_4$ must be to the right of both $a_2$ and $a_3$}. This is true because $(a_3,a_4)$ is a forward edge which implies $\fav(a_4)$ is to the right of $\fav(a_3)$, $\fav(a_2)=\fav(a_3)$, and \cref{lem:favorite_to_left} shows that if an agent's favorite item is to the right of another agent's favorite item then that agent is also to the right of the other agent.

    \cref{lem:opt_structure} shows that matching the leftmost agent to the leftmost item is optimal for both the egalitarian and utilitarian social cost. $a_1$ is to the left of $a_3$, and $M^*(a_2)$ is to the left of $M^*(a_4)$, so if we were to look locally at $a_1$, $a_3$, $M^*(a_2)$, and $M^*(a_4)$, matching $a_1$ to $M^*(a_2)$ and $a_3$ to $M^*(a_4)$ must be optimal compared to matching $a_1$ to $M^*(a_4)$, and $a_3$ to $M^*(a_2)$. Thus, both Inequalities \eqref{ineq:max} and \eqref{ineq:sum} hold. 

    Let $m_q = \max\limits_{(a_i,a_j) \in E}^{\qquad q}(d(a_i,M^*(a_j)))$. $\C{k}(\Sg_M)$ can be written as
    \begin{equation*}
        \C{k}(\Sg_M) = \sum_{q=1}^km_q
    \end{equation*}

    Let $m_x=d(a_1,M^*(a_2))$ and $m_y=d(a_3,M^*(a_4))$. Let $m_x'=d(a_1,M^*(a_4))$ and $m_y'=d(a_3,M^*(a_2))$. We will do a case analysis on wether neither $m_x$ or $m_y$ are one of the $k$ max costs of $\Sg_M$, one of the two are in the $k$ max costs of $\Sg_M$, and both are in the $k$ max costs of $\Sg_M$.
    
    \medskip
    \noindent 
    \textbf{Case 1 [Neither $m_x$ or $m_y$ are in the $k$ max costs costs of $\Sg_M$]:} If neither $m_x'$ or $m_y'$ are in the $k$ max costs of $\Sg_{M'}$ then $\C{k}(\Sg_M)=\C{k}(\Sg_{M'})$. If just one of the two are in the $k$ max costs of $\Sg_{M'}$, WLOG let it be $m_x'$, then $\C{k}(\Sg_{M'})=\C{k}(\Sg_M)-m_k+m_x'$, and because $m_x'$ is in the $k$ max costs, $m_k$ cannot be anymore, so $m_k \leq m_x'$ and thus $\C{k}(\Sg_M) \leq \C{k}(\Sg{M'})$. A similar argument shows that if both are in the $k$ max costs then $\C{k}(\Sg_M) \leq \C{k}(\Sg_{M'})$, with the difference being both $m_k$ and $m_{k-1}$ are no longer in the $k$ max costs of $\Sg_{M'}$.

     \medskip
    \noindent 
    \textbf{Case 2 [One of $m_x$ or $m_y$ are in the $k$ max costs of $\Sg_M$]:} Without loss of generality, let $m_x$ be the max of the two. Because $\max(m_x,m_y) = m_x \leq \max(m_x',m_y')$, then either $m_x'$ or $m_y'$ is larger than $m_x$ and so at least one of the two are in the $k$ max costs of $\Sg_{M'}$. If just one, then $\C{k}(\Sg_{M'}) = \C{k}(\Sg_M) - m_x + \max(m_x',m_y')$ and thus $\C{k}(\Sg_M) \leq \C{k}(\Sg_{M'})$. If both, then  $\C{k}(\Sg_{M'}) = \C{k}(\Sg_M) - m_x - m_k + m_x' + m_y'$, at least one of the two are guaranteed to be greater than $m_x$, and both are greater than $m_k$ (for the same reason as in case 1), so $\C{k}(\Sg_M) \leq \C{k}(\Sg_{M'})$.

     \medskip
    \noindent 
     \textbf{Case 3 [Both $m_x$ and $m_y$ are in the $k$ max costs of $\Sg_M$]:} Again, because $\max(m_x,m_y) \leq \max(m_x',m_y')$ at least one of the two are in the $k$ max costs of $\Sg_{M'}$. If just one of the two are then $\C{k}(\Sg_{M'}) = \C{k}(\Sg_M) - m_x - m_y + \max(m_x',m_y') + m_{k+1}$. As shown above, $m_x + m_y \leq m_x' + m_y'$, and $\min(m_x',m_y') \leq m_{k+1}$ (because $m_{k+1}$ is one of the $k$ max costs of $\Sg_{M'}$ while $\min(m_x',m_y')$ is not) so $m_x' + m_y' \leq \max(m_x',m_y') + m_{k+1}$, thus $m_x + m_y \leq \max(m_x',m_y') + m_{k+1}$ which implies $\C{k}(\Sg_M) \leq \C{k}(\Sg_{M'})$. If both are in the $k$ max costs of $\Sg_{M'}$ then $\C{k}(\Sg_{M'}) = \C{k}(\Sg_M) - m_x - m_y + m_x' + m_y'$, and thus $m_x + m_y \leq m_x' + m_y'$ implies that $\C{k}(\Sg_M) \leq \C{k}(\Sg_{M'})$. 
\end{proof}

\begin{proof}[Proof of \cref{lem:forward-edges}]
Applying \cref{algname:remove_forward_edges} to $\Sg_M$, where $M$ is a matching returned by \cref{algname:order_match}, gives $\Sg_{M'}$. \cref{lem:rfe_agents} shows that while there exists a forward edge $(a_3,a_4) \in E$, there also exists an edge $(a_1,a_2) \in E$ that satisfies the conditions of \cref{line:conditions} of \cref{algname:remove_forward_edges}, so the while loop of \cref{algname:remove_forward_edges} is always able to execute. \cref{lem:rfe_terminate} shows that \cref{algname:remove_forward_edges} terminates and will return permutation graph $\Sg_{M'}$. Given that there is a while loop in \cref{algname:remove_forward_edges} that continues to iterate while there is a forward edge in $E$, \cref{algname:remove_forward_edges} terminating and returning $\Sg_{M'}$ implies that $\Sg_{M'}$ contains no forward edges. \cref{lem:back-edges} shows that $\Sg_M$ contains no backward edges, so \cref{lem:rfe_backward} implies that $\Sg_{M'}$ contains no backward edges. Finally, \cref{lem:rfe_cost} implies that $\C{k}(M) \leq  \C{k}(M')$ for any $k \in [n]$.
\end{proof}

\paragraph{Bounding the cost of all remaining edges.} \ \\
Having addressed forward and backward edges, we now provide an upper bound for the cost of all remaining edges.

In order to bound the cost of each of the remaining edges, we first prove the following lemma regarding the leftmost and the rightmost agents, $a_\ell^*$ and $a_r^*$, which we subsequently use in the proof of \cref{lem:upper-bound-edges}.

\begin{lemma}\label{lem:star_max_gin}
Let $a_\ell^*$ and $a_r^*$ be the true leftmost and rightmost agents. For all $g \in \Gout$, we have $g_r \succ_{a_\ell^*} g$ and $g_\ell \succ_{a_r^*} g$.
\end{lemma}

\begin{proof}
    We will prove this statement by proving that $(a_\ell^*, a_r^*) \in \arg\max_{a_i, a_j} |\Gin(a_i, a_j)|$ which implies that the set $\Gin$ defined in the first step of \cref{algname:order_match} is the set $\Gin(a_\ell^*,a_r^*)$. The definition of the set $\Gin(a_\ell^*,a_r^*)$ implies that for all $g \in \Gout$, $g_r \succ_{a_\ell^*} g$ and $g_\ell \succ_{a_r^*} g$

    If $a_\ell^*$ is to the right of $g_r$, then $g_r$ would be the favorite item of all agents because $g_r$ is the rightmost item with positive plurality and $a_\ell^*$ would be in between all the other agents and $g_r$, which would result in \cref{algname:order_match} defining $g_\ell = g_r$ and $\Gin$ would be empty no matter which agents are used to define it and the proof is complete. For the case that $a_\ell^*$ is to the left of $g_r$, we will show that for any $a_\ell'$ such that $\fav(a_\ell') = g_\ell$, $|\Gin(a_\ell^*, a_r^*)| \geq |\Gin(a_\ell', a_r^*)|$. We will do this by showing that all items that are in $\Gin(a_\ell', a_r^*)$ are also in $\Gin(a_\ell^*, a_r^*)$. $a_\ell'$ must be to the left of $g_r$ because $\fav(a_\ell')=g_\ell$, $g_\ell \neq g_r$, and if $a_\ell'$ were to the right of $g_r$, they would have $g_r$ as their favorite item because $g_r$ is the rightmost item with positive plurality. $a_\ell'$ is to the right of $a_\ell^*$ because $a_\ell^*$ is the leftmost agent. All items in between $a_\ell^*$ and $g_r$ are included in $\Gin(a_\ell^*, a_r^*)$, and similarly all items in between $a_\ell'$ and $g_r$ are included in $\Gin(a_\ell', a_r^*)$. Because $a_\ell'$ is to the right of $a_\ell^*$, and both are to the left of $g_r$, any item in between $a_\ell'$ and $g_r$ is also in between $a_\ell^*$ and $g_r$ and thus also in $\Gin(a_\ell^*, a_r^*)$. 
    
    As for any item that is to the left of $a_\ell'$ that is included in $\Gin(a_\ell', a_r^*)$, let this item be $g_\ell'$, if $g_\ell'$ is to the right of $a_\ell^*$ then it is in between $a_\ell^*$ and $g_r$ and thus also in $\Gin(a_\ell^*, a_r^*)$. If $g_\ell'$ is to the left $a_\ell^*$ then it is closer to $a_\ell^*$ than $a_\ell'$, i.e., $d(a_\ell^*, g_\ell') \leq d(a_\ell', g_\ell')$. Because $a_\ell'$ is to the right of $a_\ell^*$, and both are to the left of $g_r$, then $a_\ell'$ is closer to to $g_r$ than $a_\ell^*$ is, i.e., $d(a_\ell',g_r) \leq d(a_\ell^*,g_r)$. Because $g_\ell' \in \Gin(a_\ell', a_r^*)$ then $d(a_\ell', g_\ell') \leq d(a_\ell',g_r)$. Thus $d(a_\ell^*, g_\ell') \leq d(a_\ell', g_\ell') \leq d(a_\ell',g_r) \leq d(a_\ell^*,g_r)$, which implies $d(a_\ell^*,g_\ell') \leq d(a_\ell^*,g_r)$, which implies $g_\ell' \in \Gin(a_\ell^*, a_r^*)$. It has been shown that every item in $\Gin(a_\ell', a_r^*)$ is also in $\Gin(a_\ell^*, a_r^*)$, thus $|\Gin(a_\ell^*, a_r^*)| \geq |\Gin(a_\ell', a_r^*)|$. A symmetric argument can show that $a_r^*$ maximizes the size of $\Gin$ for a fixed $a_\ell^*$.
\end{proof}

\begin{lemma}\label{lem:upper-bound-edges}
    Consider any permutation-graph $\Sg_M = (A,E)$ without forward or backward edges.
    For any $(a_i,a_j) \in E$,
    
    $$d(a_i,M^*(a_j)) \leq d(a_i,M^*(a_i)) + 2 \cdot d(a_j,M^*(a_j)).$$
\end{lemma}
\begin{proof}
    Given that $(a_i,a_j)$ is neither a forward or backward edge, then $(a_i,a_j)$ is either an internal or inward edge. It could also be the case that $a_i \in \Ain$ and $a_j \in \Aout$, or both $a_i,a_j \in \Aout$. We prove that the claimed inequality works for each case separately.

    \medskip
    \noindent
    \textbf{Case 1 [$(a_i,a_j)$ is an internal edge]:} Because $(a_i,a_j)$ is an internal edge, both agents share the same favorite item, say $g^*$. Using the triangle inequality and the fact that $d(a,g^*) \leq d(a,g)$ for any $a \in \{a_i,a_j\}$ and $g \in G$, we have:
    \begin{align*}
        d(a_i,M^*(a_j)) &\leq d(a_i,g^*) + d(a_j,g^*) + d(a_j,M^*(a_j))\\
        &\leq d(a_i,M^*(a_i)) + d(a_j,M^*(a_j)) + d(a_j, M^*(a_j))\\
        &= d(a_i,M^*(a_i)) + 2 \cdot d(a_j,M^*(a_j)).      
    \end{align*}
    
    \medskip
    \noindent
    \textbf{Case 2 [$(a_i,a_j)$ is an inward edge]:} Without loss of generality, let $a_i \in \Aout^\ell$, and thus $M^*(a_i) \in \Gout^\ell$.
    If $a_j$ is to the right of $g_r$, then we have:
    \begin{align*}
        d(a_i,M^*(a_j)) &\leq d(a_i,g_r) + d(a_j, g_r) + d(a_j,M^*(a_j))\\
        &\leq d(a_i,M^*(a_i)) + d(a_j,M^*(a_j)) + d(a_j, M^*(a_j))\\
        &= d(a_i,M^*(a_i)) + 2 \cdot d(a_j,M^*(a_j)).
    \end{align*}
    The first inequality holds due to the triangle inequality. For the second inequality, since $a_j$ is to the right of $g_r$, $g_r$ must be the favorite item of $a_i$ since it is the rightmost item with positive plurality, and thus $d(a_j,g_r) \leq d(a_j,g)$ for any $g \in G$. Similarly, if $a_i$ is to the right of $g_r$, $d(a_i,g_r) \leq d(a_i,g)$ for any $g \in G$. Otherwise, since $a_\ell^*$ (the leftmost agent) prefers $g_r$ over all items in $\Gout$ (by \cref{lem:star_max_gin}) and $a_i$ is to the right of $a_\ell^*$, $a_i$ 
    must be between $a_\ell^*$ and $g_r$; thus, $d(a_i,g_r) \leq d(a_i,M^*(a_i))$.
    
   Otherwise, $a_j$ is to the left of $g_r$. 
   If $M^*(a_j)$ is to the right of $g_r$, then $a_j$ prefers $g_r$ to $M^*(a_j)$, so the above inequalities still hold. If $M^*(a_j)$ is to the left of $g_r$, then $a_i$ must prefer $M^*(a_j)$ to $M^*(a_i)$ because either $M^*(a_j)$ is in between $a_i$ and $g_r$ and $a_i$ prefers $g_r$ to $M^*(a_i)$ (as shown above), or $M^*(a_j)$ is to the left of $a_i$, but must still be to the right of $M^*(a_i)$ because all items in $\Gin$ are to the right of items in $\Gout^\ell$. Hence, we have:
    \begin{align*}
        d(a_i,M^*(a_j)) &\leq d(a_i,M^*(a_i))\\
        &\leq d(a_i,M^*(a_i)) + d(a_j,M^*(a_j)) + d(a_j, M^*(a_j))\\
        &= d(a_i,M^*(a_i)) + 2 \cdot d(a_j,M^*(a_j)).
    \end{align*}

    \medskip
    \noindent
    \textbf{Case 3 [$a_i \in \Ain$ and $a_j \in \Aout$]:} 
    Without loss of generality, let $a_j \in \Aout^\ell$, and thus $M^*(a_j) \in \Gout^\ell$. First, let $d(a_i,M^*(a_i)) < d(a_i,g_r)$. For this to be the case, $M^*(a_i)$ must be to the left of $g_r$ because if it were to the right of $g_r$ then either $a_i$ is to the left of $g_r$ (and thus prefers $g_r$ because $g_r$ is between $a_i$ and $M^*(a_i)$), or $a_i$ is to the right of $g_r$ (and thus $g_r$ must be their favorite item because $g_r$ is the rightmost item with positive plurality). We have:
    \begin{align*}
        d(a_i,M^*(a_j)) &\leq d(a_i,M^*(a_i)) + d(a_j, M^*(a_i)) + d(a_j, M^*(a_j))\\
        &\leq d(a_i,M^*(a_i)) + d(a_j,M^*(a_j)) + d(a_j, M^*(a_j))\\
        &= d(a_i,M^*(a_i)) + 2 \cdot d(a_j,M^*(a_j)).   
    \end{align*}
    The first inequality holds due to the triangle inequality. For the second inequality, if $M^*(a_i)$ is to the left of $a_j$, observe that $M^*(a_i)$ is to the right of $M^*(a_j)$ because all items in $\Gin$ are to the right of all items in $\Gout^\ell$, so $a_j$ must prefer $M^*(a_i)$ over $M^*(a_j)$ because it is between $a_j$ and $M^*(a_j)$. If $a_j$ is to the right of $g_r$, it must be the case that $M^*(a_i)$ is to the left of $a_j$. If $a_j$ is to the left of $g_r$, observe that $a_\ell^*$ (the leftmost agent) prefers $g_r$ to all items in $\Gout$ (by \cref{lem:star_max_gin}). Since $a_j$ is to the right of $a_\ell^*$, $a_j$ is between $a_\ell^*$ and $g_r$, and thus prefers $g_r$ to any item in $\Gout^\ell$. Therefore, if $M^*(a_i)$ is to the right of $a_j$, $a_j$ must prefer it to $M^*(a_j)$, because $a_j$ prefers $g_r$ (which is even farther to the right as shown above) to $M^*(a_j)$. Thus $d(a_j,M^*(a_i)) \leq d(a_j,M^*(a_j))$.

    Second, let $d(a_i,g_r) \leq d(a_i,M^*(a_i))$. We have:
    \begin{align*}
        d(a_i,M^*(a_j)) &\leq d(a_i,g_r) + d(a_j, g_r) + d(a_j, M^*(a_j))\\
        &\leq d(a_i,M^*(a_i)) + d(a_j,M^*(a_j)) + d(a_j, M^*(a_j))\\
        &= d(a_i,M^*(a_i)) + 2 \cdot d(a_j,M^*(a_j)).     
    \end{align*}
    The first inequality holds due to the triangle inequality. For the second inequality, if $a_j$ is to the right of $g_r$, $g_r$ must be $a_j$'s favorite item because $g_r$ is the rightmost item with positive plurality. If $a_j$ is to the left of $g_r$, $a_\ell^*$ (the leftmost agent) preferring $g_r$ to all items in $\Gout$ implies $a_j$ is closer to $g_r$ then any item in $\Gout^\ell$ because $a_j$ is to the right of $a_\ell^*$, and thus in between $a_\ell^*$ and $g_r$. Thus, overall, $d(a_j,g_r) \leq d(a_j,g)$ for any $g \in \Gout^\ell$.

    \medskip
    \noindent
    \textbf{Case 4 [$a_i,a_j \in \Aout$]:} Note that $M^*(a_i),M^*(a_j) \in \Gout$ because both $a_i,a_j \in \Aout$. We have:
    \begin{align*}
        d(a_i,M^*(a_j)) &\leq d(a_i,g_\ell) + d(a_j,g_\ell) + d(a_j,M^*(a_j))\\
        &\leq d(a_i,M^*(a_i)) + d(a_j,M^*(a_j)) + d(a_j, M^*(a_j))\\
        &= d(a_i,M^*(a_i)) + 2 \cdot d(a_j,M^*(a_j)).      
    \end{align*}
    The first inequality holds due to the triangle inequality. For the second inequality, first observe that $a_r^*$ (the rightmost agent) prefers $g_\ell$ to all items in $\Gout$ (by \cref{lem:star_max_gin}). Since $a_i$ and $a_j$ are to the left of $a_r^*$, they are closer to $g_\ell$ than any item
    in $\Gout^r$. If these agents are to the left of $g_\ell$, then $g_\ell$ must be their favorite item since $g_\ell$ is the leftmost item with positive plurality. So, overall, $d(a,g_\ell) \leq d(a,g)$ for any $a \in \{a_i,a_j\}$ and $g \in \Gout$.
\end{proof}

We can now prove the main result of this section.

\begin{proof}[Proof of \cref{thm:distortion(UB}]
By \cref{lem:back-edges}, the permutation-graph of the matching returned by \cref{algname:order_match} has no backward edges, and by \cref{lem:forward-edges}, the $k$-centrum cost of this matching can be upper-bounded by the $k$-centrum cost of another matching whose permutation-graph has no forward or backward edges. 
We now show that for any matching $M$ whose permutation-graph $\Sg_M = (A,E)$ contains no forward or backward edges, we have $\C{k}(M) = \C{k}(\Sg_M) \leq 3 \cdot \C{k}(M^*)$ for all $k \in [n]$. Specifically, if we let $E_k\subseteq E$ be the $k$ edges $(a_i, a_j)$ of $E$ with the largest $d(a_i, M^*(a_j))$ cost, then 
    \begin{align*}
        \C{k}(M)
        &= \sum_{(a_i, a_j)\in E_k} d(a_i,M^*(a_j))\\
        &\leq \sum_{(a_i, a_j)\in E_k} \bigg( d(a_i,M^*(a_i))+ 2 \cdot d(a_j,M^*(a_j)) \bigg)\\
        &= \sum_{a_i : (a_i, a_j)\in E_k}  d(a_i,M^*(a_i))+  \sum_{a_j : (a_i, a_j)\in E_k} 2 \cdot d(a_j,M^*(a_j))\\
        &\leq 3 \cdot \sum_{q=1}^k{\max\limits_{a \in A}}^q d(a,M^*(a)) \\
        &= 3 \cdot \C{k}(M^*),
    \end{align*}
where the first inequality uses \cref{lem:upper-bound-edges} for every edge in $ E_k$ and the subsequent equality breaks down the contribution of each edge $(a_i, a_j)\in E_k$ to the contribution of $a_i$ and the contribution of $a_j$. The final inequality uses the fact that: \emph{(i)} all nodes have in- and out-degree $1$ in any permutation-graph, so each agent $a$ participates in at most two edges in $E_k$, $(a, a_j)$ and $(a_i, a)$, and \emph{(ii)} the total cost of the $k$ agents $a_i: (a_i, a_j)\in E_k$  (and that of the $k$ agents $a_j: (a_i, a_j)\in E_k$) is at most the total cost of the $k$ maximum-cost agents.
\end{proof}

\section{Two-sided case: Achieving the optimal matching}
In this section we consider a two-sided matching, in which instead of $n$ agents that have preferences of $n$ items, we have a set of $n$ {\em takers} denoted by $A$ and a set of $n$ {\em givers} denoted by $B$. We have access to the ranking of each taker over the givers, and of each giver over the takers. We will refer to each of these sets as a side of the matching problem, and will use the term {\em agent} to refer to a taker or a giver.

We will show that using the ordinal rankings of the takers and the givers, we can compute an optimal matching. In particular, to do so, it suffices to show that we can recover the true ordering of the takers and the giver on the line. Then, we can greedily compute an optimal matching using \Cref{lem:opt_structure} by considering one side as the items.

\begin{theorem}
\label{thm:two-sided}
    For the two-sided matching problem on a line metric, we can compute an optimal matching with respect to the $k$-centrum cost $\C{k}$ for any $k\in [n]$.
\end{theorem}

\begin{proof}
    We first show that we can recover the true relative ordering of the takers as well as the true relative ordering of the givers. We do so for the takers; the argument for the givers is symmetric. 
    We consider two cases: Either all takers agree on their least-preferred giver, or there are two givers that appear at the bottom of the rankings of the takers.

    In the first case, let $b$ be the giver that appears at the bottom of every taker's ranking. We claim that all takers are on the same side of $b$. If not, then there is a taker $i$ that is to the left of $b$ and a taker $j$ that is to the right of $b$. This implies that $i$ weakly prefers $b$ to $\fav(j)$, which contradicts the fact that $b$ is the least-preferred giver of all takers. Now, since all takers are on the same side of $b$, $\succ_b$ is the true ranking of the takers.

    In the second case, let $b_r$ and $b_\ell$ be the two extreme givers, and for the ease of exposition assume that $b_r$ is to the right of $b_\ell$. First, we claim that there is no giver to the left of $b_\ell$, and similarly to the right of $b_r$. For the sake of contradiction, assume that there is a giver $b$ to the left of $b_\ell$, and let $i$ be a taker with $b_r \succ_i b_\ell$. This means that this taker is to the right of $b_\ell$, and hence $b_r \succ_i b_\ell \succ_i b$, which contradicts the fact that only $b_r$ and $b_\ell$ appear at the bottom of the takers' rankings.

Let $A^{b_\ell}$ be the set of takers that have $b_\ell$ as their top choice.
Since $b_\ell$ is the leftmost giver, all the takers to the left of it are in $A^{b_\ell}$. Hence, all the takers in $A \setminus A^{b_\ell}$ are to the right of $b_\ell$ and to the right of all the takers in $A^{b_\ell}$. Consequently, the ordering $\succ_{b_\ell}$ over $A \setminus A^{b_\ell}$ is the true ordering of these takers. Similarly, $\succ_{b_r}$ recovers the true ordering of $A^{b_\ell}$. By putting these two together, we have the true ordering of the whole set of takers.

Now that we have the true relative ordering of the takers and the givers, we can greedily match them to get the optimal matching $M^*$. The only issue is that these orderings are correct up to reversal and can be oriented in opposite directions. In other words, one ordering can be from left to right and the other from right to left. So, we have to make them consistent. We consider two cases depending on if at least one of the two sides have two different agents at the bottom of their lists (case 1), or both sides have the same (case 2).

\smallskip
\noindent 
{\bf Case 1.} Assume that the takers have two different givers as their least preferred. Let $a$ be an extreme taker (given the ordering of the takers, which we now know), and assume that $a$ is the rightmost one. Then, $a$'s least preferred giver must be the leftmost one. We use this information to give the same direction to the relative orderings of the takers and the givers. 

\smallskip
\noindent 
{\bf Case 2.} Here, each side has the same least-preferred agent in the other side.  
Let $a$ be an extreme taker (according to the ordering of the takers), and assume that $a$ is to the rightmost one. Then, $a$ is to the same side of all the givers, and thus $\fav(a)$ has to be one of the extreme takers. Hence, we order the givers so that $\fav(a)$ is the rightmost one. 

Now that the orderings are in the same direction we can greedily match them from left to right, thus computing the optimal matching $M^*$.
\end{proof}

\subsection{Query Model}
We now consider a setting where we are given some prior information about the preferences of the agents, and we can obtain further information by making specific types of {\em queries}. In particular, we consider two types of prior information: 
\begin{itemize}
    \item {\em One-sided setting}: We are given the complete ordinal preferences of the takers in $A$ over the givers in $B$. 
    \item {\em Zero-knowledge setting}: No initial information about the preferences of the agents is given.
\end{itemize} 
Similarly, we consider two types of queries that we can make to gain more information:
\begin{itemize}
    \item {\em Full-preference queries}: We ask an agent to reveal its ordinal preference over all the agents on the other side. 
    \item {\em Rank queries}: We ask an agent to reveal its $t$-th preferred agent on the other side. 
\end{itemize}
For each combination of prior information and type of queries as defined above, we are interested in pinpointing the number of queries that are sufficient to compute an optimal matching. Given that an optimal matching can be computed for the two-sided matching problem as we showed above, we are essentially aiming to find out how many queries we need to make to reveal the true ordering of the agents of each side on the line. 

\begin{figure*}[t]
\begin{subfigure}[t]{\textwidth}
    \centering
    \scalebox{0.85}{
    \begin{tikzpicture}[node distance={13mm}, thick, c/.style = {draw, circle, minimum size=2mm}, s/.style = {draw, rectangle, minimum size=2mm}, t/.style = {draw, regular polygon, regular polygon sides=3, minimum size=2mm]}, invisible/.style = {minimum size=0mm, inner sep=0mm, text width=0mm, outer sep=0mm, draw=none}]
    \node[c] [label=below:$b_n$] (1) {};
    \node[c] [right of=1, label=below:$b_{n-1}$] (2) {};
    \node[c] [right of=2, label=below:$b_{i-1}$] (i) {};
    \node[invisible] [right of=i] (ii) {};
    \node[c] [right of=ii, label=below:$b_{i+1}$] (3) {};
    \node[c] [right of=3, label=below:$b_2$] (4) {};
    \node[c] [right of=4, label=below:$b_1$] (5) {};
    \node[c] [right of=5, label=below:$a_1$] (6) {};
    \node[c] [right of=6, label=below:$a_2$] (7) {};
    \node[c] [right of=7, label=below:$a_{n-1}$] (8) {};
    \node[c] [right of=8, label=below:$a_{n}$] (9) {};
    \node[invisible] [right of=9] (10) {};
    \node[c] [right of=10, label=below:$b_{i}$] (11) {};
    \draw (1) -- node[midway, above] {$n$} (2);
    \draw (4) -- node[midway, above] {$n$} (5);
    \draw[draw=none] (2) -- node[midway] {$\cdots$} (i);
    \draw (i) -- node[midway, above] {$2n$} (3);
    \draw[draw=none] (3) -- node[midway] {$\cdots$} (4);
    \draw (5) -- node[midway, above] {$n$} (6);
    \draw (6) -- node[midway, above] {$1$} (7);
    \draw[draw=none] (7) -- node[midway] {$\cdots$} (8);
    \draw (8) -- node[midway, above] {$1$} (9);
    \draw (9) -- node[midway, above] {$ni$} (11);
    \end{tikzpicture}
    }
\caption{Top metric.}
\label{fig:topmetric}
\end{subfigure}

\begin{subfigure}[t]{\textwidth}
    \centering
    \scalebox{0.85}{
    \begin{tikzpicture}[node distance={13mm}, thick, c/.style = {draw, circle, minimum size=2mm}, s/.style = {draw, rectangle, minimum size=2mm}, t/.style = {draw, regular polygon, regular polygon sides=3, minimum size=2mm]}, invisible/.style = {minimum size=0mm, inner sep=0mm, text width=0mm, outer sep=0mm, draw=none}]
    \node[c] [label=below:$b_n$] (1) {};
    \node[c] [right of=1, label=below:$b_{n-1}$] (2) {};
    \node[c] [right of=2, label=below:$b_{j-1}$] (i) {};
    \node[invisible] [right of=i] (ii) {};
    \node[c] [right of=ii, label=below:$b_{j+1}$] (3) {};
    \node[c] [right of=3, label=below:$b_2$] (4) {};
    \node[c] [right of=4, label=below:$b_1$] (5) {};
    \node[c] [right of=5, label=below:$a_1$] (6) {};
    \node[c] [right of=6, label=below:$a_2$] (7) {};
    \node[c] [right of=7, label=below:$a_{n-1}$] (8) {};
    \node[c] [right of=8, label=below:$a_{n}$] (9) {};
    \node[invisible] [right of=9] (10) {};
    \node[c] [right of=10, label=below:$b_{j}$] (11) {};
    \draw (1) -- node[midway, above] {$n$} (2);
    \draw (4) -- node[midway, above] {$n$} (5);
    \draw[draw=none] (2) -- node[midway] {$\cdots$} (i);
    \draw (i) -- node[midway, above] {$2n$} (3);
    \draw[draw=none] (3) -- node[midway] {$\cdots$} (4);
    \draw (5) -- node[midway, above] {$n$} (6);
    \draw (6) -- node[midway, above] {$1$} (7);
    \draw[draw=none] (7) -- node[midway] {$\cdots$} (8);
    \draw (8) -- node[midway, above] {$1$} (9);
    \draw (9) -- node[midway, above] {$nj$} (11);
    \end{tikzpicture}
    }
\caption{Bottom metric.}
\label{fig:bottommetric}
\end{subfigure}
\caption{The two metrics that establish the lower bound shown in \cref{thm:fullprefLB}.}
\label{fig:ijmetrics}
\end{figure*}
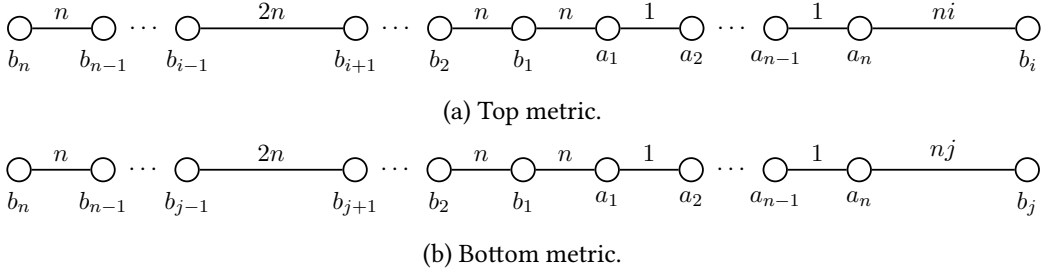

In the one-sided setting, $n$ full-preference queries are clearly sufficient to compute the optimal matching as we would then have the ordinal preference of all agents, leading to the two-sided matching problem. We further show that this amount of such queries is necessary. 

\begin{theorem}
\label{thm:fullprefLB}
    In the one-sided setting, at least $n-1$ full-preference queries are required to compute an optimal matching.
\end{theorem}

\begin{proof}
    Consider the instance where all the takers have the same ranking $b_1 \succ b_2 \succ \cdots \succ b_n$ over the givers. Assume that an algorithm makes strictly less than $n-1$ queries to the givers, and all queried givers return the same ranking $a_1 \succ a_2 \succ \cdots \succ a_n$. Let $b_i$ and $b_j$ be two givers that have not been queries, and consider the two metrics illustrated in \Cref{fig:ijmetrics}. 
    We first argue that these metrics are to reported preferences. Since they are symmetric, we only check the top metric. On one hand since all queried givers are to the left of all the takers, the preferences of the givers are consistent with this metric. On the other hand, the distance of giver $\ell \neq i$ to taker $k$ is exactly $\ell n+k-1$, and the distance of $b_i$ to $a_k$ is $ni+(n-k)$, which is at least $ni$ and at most $ni+n-1$; the latter is less than the distance of all the agents to $b_{i-1}$ and more than their distance to $b_{i+1}$, and hence the preferences of the takers are consistent with the metric.
    
    Now observe that the optimal matching in the top metric is such that $a_n$ is matched to $b_i$. In contrast, any optimal matching in the bottom metric must be such that $a_n$ should be matched to $b_j$. Hence, it is not possible to find the optimal matching with less than $n-1$ queries.
\end{proof}

Due to \cref{thm:fullprefLB}, since $n-1$ full-preference queries are required even when the whole ordering of one side is given, by applying the argument to each side independently, we obtain a lower bound of $2n-2$ on the number of required full-preference queries for the zero-knowledge setting, which is again tight as $2n$ full-preference queries reveal the orderings of both sides of agents. 

\begin{corollary}
    In the zero-knowledge setting, we need at least $2n-2$ full-preference queries to find the optimal matching.
\end{corollary}

Since we can simulate a full-preference query using at most $n-1$ rank queries, we can clearly compute an optimal matching by using at most $n(n-1)$ rank queries in the one-sided setting. In the following theorem, we show that we can do this using only $3n-4$ rank queries.

\begin{theorem}\label{thm:one-sided-rank-queries}
    In the one-sided setting, we can compute an optimal matching by making $3n-4$ rank queries.
\end{theorem}
\begin{proof}
    If two givers appear as the bottom choice of the takers, we can rank the takers who prefer the first one based on the preference of the second one and vice versa. This requires $2n-2$ rank queries. Now, let $a$ be an extreme taker. We query all the givers for their top takers and partition them based on whether they prefer $a$ to other takers or not. The next step is to use the ranking of $a$ to sort the givers who do not have $a$ as their top choice, and another taker to sort the givers who has $a$ as its top choice. This requires $3n-4$ queries in total.

    If only one giver appears as the bottom choice, say $b$, then we know that all the takers are to the same side of $b$. So, we can use the ranking of $b$ as the true ordering of the takers (this required $n-1$ rank queries). After that, similarly to the previous case, we query the givers for their top choice and partition them based on that. Then, we use the rankings of one extreme taker and one more taker to find the true ordering of the givers. This requires $2n-2$ rank queries in total.
\end{proof}

For the zero-knowledge setting, we show that only a constant factor more rank queries are required to compute the optimal matching. 

\begin{theorem}\label{thm:zero-knowledge-rank-queries}
    In the zero-knowledge setting, we can find the optimal matching with $5n-4$ rank queries.
\end{theorem}
\begin{proof}
Start with an arbitrary taker $a_0$ and query its least-preferred giver; let this be $b_1$.
Without loss of generality, we assume that $b_1$ is to the right of $a_0$. 
We then query the top choice of all other takers to check if $b_1$ is the top choice of any of them (so, we have made $n$ queries so far). 
We now switch between two cases depending on whether there exists a taker $a_1$ such that $\fav(a_1) = b_1$ (Case 1), or not (Case 2).

\medskip
\noindent 
{\bf Case 1:} Since $b_1$ is the least preferred choice of $a_0$, there must be no giver to the right of $b_1$. In addition, since $b_1$ is the favorite of $a_1$, there must be no giver between these two. 
This implies that all the givers other than $b_1$ are located to the left of $a_1$. 
Hence, we can use the ranking of $a_1$ to retrieve the true ordering of all the givers (note that we have made $2n-2$ queries so far). Now we partition the takers into two sets, the one who have $b_1$ as their top choice and the ones that do not have him as their top choice. We know that all the members of the first set are to the right of all the members of the second set by \Cref{lem:favorite_to_left}. We use the ranking of $b_1$ to find the ordering of the members of the second set and another giver to rank the members of the first set. Putting these two rankings together we have the true ordering of all the takers and givers and can find the optimal matching.

\medskip
\noindent 
{\bf Case 2:} If there exists a taker to the right of $b_1$, then $b_1$ would be its favorite giver. Since none of the takers report $b_1$ as their favorite, we can conclude that all the takers are to the left of $b_1$. This means that the ranking of $b_1$ is the true ordering of the takers and we can retrieve this ordering with $n-1$ extra queries. After that we can check the top taker of all the givers, and do as we did in the last part of the former case. Here we do not need to query $b_1$ since we know it is the rightmost giver. This requires  $n-1$ queries for the top choices of all the takes, and $2n-2$ queries to retrieve the remaining full rankings of two of the takers. In total, we need $5n-4$ queries.
\end{proof}

\section{Conclusion and Future Directions}
We showed a tight bound of $3$ on the distortion of ordinal algorithms for the metric one-sided matching problem on the line with respect to the $k$-centrum cost, the sum of the $k$-largest costs among all $n$ agents, for any $k \in [n]$. We further showed that, for the two-sided matching problem, the ordinal preferences of the agents are sufficient to reveal their true ordering on the line which can then be used to compute an optimal matching; in fact, even less information is sufficient to achieve optimality. 

Although we have resolved the distortion of matching on the line, there are several natural directions for future research. For example, it would be interesting to investigate the one-sided matching problem on the line using the learning-augmented framework \citep{lykouris2021competitive}. According to this framework, which was recently also used for the metric distortion of single-winner voting \citep{BFGT24}, the designer is provided with some (potentially inaccurate) {\em prediction} regarding the agents' cardinal preferences. The goal is to achieve an improved distortion whenever the prediction is accurate (known as consistency) while maintaining a good distortion even if the prediction is arbitrarily inaccurate (known as robustness). In our setting, the main question would be whether a consistency better than $3$ can be combined with a robustness of $3$.

For more general metrics, the distortion of the one-sided matching problem still remains wide open, and bridging the gap between the known lower bound of $\Omega(\log{n})$ and the upper bound of $O(n^2)$ for deterministic algorithms or the upper bound of $O(n)$ for randomized algorithms is quite possibly one of the most challenging open questions in the distortion literature at the moment. To obtain positive results, as we did in this paper for the line metric, one could next consider other specific metric spaces, such $2$-dimensional Euclidean spaces or trees (on which the lower bound holds). Additionally, the two-sided matching problem has not been considered before our work under the metric distortion framework, and it is thus a natural problem to study for general metrics, especially as it seems to be somewhat easier than the one-sided variant (at least for the line, as we have shown). 

\section*{Acknowledgments}
Aris Filos-Ratsikas and Mohamad Latifian were supported by the UK Engineering and Physical Sciences Research Council (EPSRC) grant
EP/Y003624/1. Vasilis Gkatzelis and Emma Rewinski were supported by NSF CAREER award CCF-2047907 and NSF grant CCF-2210502.

\bibliographystyle{plainnat}
\bibliography{references}

\newpage
\section*{Appendix}
\subsection*{The significance of the tie-breaking rule in the choice of $a_\ell$ and $a_r$ by \cref{algname:order_match}}
\begin{observation}\label{obs:tiebreak}
If \cref{algname:order_match} were to select $a_\ell$ and $a_r$ arbitrarily among agents whose favorite item is $g_\ell$ and $g_r$, respectively, then it would not achieve the optimal distortion of $3$ with respect to $\C{k}$ for $k \in [n]$. Specifically, when $k=1$ the distortion would be at least $5$, and when $k \geq 2$ the distortion would be at least $7$. 
See the metrics illustrated in Figure~\ref{fig:tiebreak}.  
\end{observation}
\begin{proof}
Throughout this proof we will refer to selecting $a_\ell$ and $a_r$ such that $|\Gin|$ is maximized as \textit{tie-breaking}. First, we will provide an instance and corresponding metric space such that when $k=1$, if there was no tie-breaking the distortion of \cref{algname:order_match} would be at least $5$. Afterwards, we will provide another instance and corresponding metric space such that when $k \geq 2$, if there was no tie-breaking the distortion of \cref{algname:order_match} would be at least $7$.

For $k=1$, consider an instance with $n$ agents $\{a_1,\ldots,a_n\}$ and $n$ items $\{g_1,\ldots,g_n\}$. The ordinal profile is such that agent $a_1$'s ranking over the items is $g_2 \succ \dots \succ g_{n-2} \succ g_1 \succ g_{n-1} \succ g_n$, agent $a_{n-1}$'s ranking over the items is $g_{n-1} \succ g_n \succ g_2 \succ \dots \succ g_{n-2} \succ g_1$, agent $a_n$'s ranking over the items is $g_{n-1} \succ g_{n} \succ g_2 \succ \dots \succ g_{n-1} \succ g_1$, and the rest of the agents rankings over the items are the same: $g_2 \succ \dots \succ g_{n-2} \succ g_{n-1} \succ g_1 \succ g_n$. Observe that only two items have positive plurality, $g_2$ and $g_{n-1}$. Agents $a_1$ through $a_{n-2}$ have $g_2$ as their favorite item, and agents $a_{n-1}$ and $a_n$ have $g_{n-1}$ as their favorite item. \cref{algname:order_match} will define these two items as $g_\ell$ and $g_r$, let $g_2=g_\ell$ and $g_{n-1}=g_r$.

Consider the metric space that is illustrated in \cref{fig:tiebreak:k=1}, where $\varepsilon > 0$ is an infinitesimal. The optimal matching consists of the pairs $(a_i,g_i)$ for $i \in [n]$, leading to a $1$-centrum cost of $d(a_1,g_1)=1$. Without tie-breaking, any agent $a_1$ through $a_{n-2}$ can be selected as $a_\ell$. Let $a_2=a_\ell$ (note $g_r \succ_{a_2} g_1$). Either agent $a_{n-1}$ or $a_n$ can be selected as $a_r$. Without loss of generality, let $a_{n-1}=a_r$. \cref{algname:order_match} would define $\Gin=\{g_2,\ldots,g_n\}$ and $\Gout=\{g_1\}$. \cref{algname:order_match} would define $\pi_g=\{g_2,\ldots,g_n\}$. In Step 2(b) there are many orderings of $\pi_a$ that \cref{algname:order_match} could define, let it define $\pi_a=\{a_1,\dots,a_n\}$. In Step 3, all agents except $a_n$ will be matched to an item in $\Gin$, leaving $a_n$ to be matched to $g_1$ in Step 4. This will result in a cost of $d(a_n,g_1)=5-5 \cdot \varepsilon$, and consequently a distortion of at least $5- 5 \cdot \varepsilon$. Notice that with tie-breaking, that only agent $a_1$ can be defined as $a_\ell$. This would lead to \cref{algname:order_match} defining $\Gin=\{g_1,\dots,g_n\}$ and $\pi_g=\{g_1,\dots,g_n\}$. Ultimately, with this definition of $\pi_g$, the largest cost \cref{algname:order_match} could achieve is $d(a,g_1)=2 - \varepsilon$ for some $a \in \{a_2, \ldots, a_{n-2}\}$.

For $k \geq 2$, consider an instance with $n$ agents $\{a_1,\ldots,a_n\}$ and $n$ items $\{g_1,\ldots,g_n\}$.  The ordinal profile is such that agent $a_{n-1}$'s ranking over the items is $g_1 \succ \dots \succ g_{n-2} \succ g_{n-1} \succ g_n$, agent $a_n$'s ranking over the items is $g_n \succ g_1 \succ \dots \succ g_{n-2} \succ g_{n-1}$, and the rest of the agents rankings over the items are the same: $g_1 \succ \dots g_{n-2} \succ g_n \succ g_{n-1}$.  Observe that only two items have positive plurality, $g_1$ and $g_n$. Agents $a_1$ through $a_{n-1}$ have $g_1$ as their favorite item, and agent $a_n$ has $g_n$ as their favorite item. \cref{algname:order_match} will define these two items as $g_\ell$ and $g_r$, let $g_1=g_\ell$ and $g_n=g_r$.

Consider the metric space that is illustrated in \cref{fig:tiebreak:kgeq2}, where $\varepsilon > 0$ is an infinitesimal. The optimal matching consists of the pairs $(a_i,g_i)$ for $i \in [n]$, leading to a $k$-centrum cost of $d(a_{n-1},g_{n-1})=1$.  Without tie-breaking, any agent $a_1$ through $a_{n-1}$ can be selected as $a_\ell$. Let $a_1=a_\ell$ (note $g_r \succ_{a_1} g_{n-1}$). $a_n$ is the only agent who has $g_r$ as their favorite item, so $a_n=g_r$. \cref{algname:order_match} would define $\Gin=\{g_1,\ldots,g_{n-2},g_n\}$ and $\Gout=\{g_{n-1}\}$.  In Step 2(b) there are many orderings of $\pi_a$ that \cref{algname:order_match} could define, let it define $\pi_a=\{a_1,\dots,a_n\}$. In Step 3, agent $a_i$ would be matched to item $g_i$ for all $i \leq n-2$, and agent $a_{n-1}$ would be matched to item $g_n$. In Step 4, agent $a_n$ would matched to item $g_{n-1}$.  This will result in a cost of $d(a_{n-1},g_n)+d(a_n,g_{n-1})=7-6 \cdot \varepsilon$, and consequently a distortion of at least $7-6 \cdot \varepsilon$, for any $k \geq 2$. Notice that with tie-breaking, that only agent $a_{n-1}$ can be defined as $a_\ell$. This would lead to \cref{algname:order_match} defining $\Gin=\{g_1,\ldots,g_n\}$ and $\pi_g=\{g_{n-1},g_1,\ldots,g_{n-2},g_n\}$. Ultimately, with this definition of $\pi_g$, the largest cost \cref{algname:order_match} could achieve is $d(a_{n-1},g)+d(a,g_{n-1})=3 - 2 \cdot \varepsilon$ for some $g \in \{g_1, \ldots, g_{n-2}\}$ and $a \in \{a_1, \ldots, a_{n-2}\}$.
\end{proof}

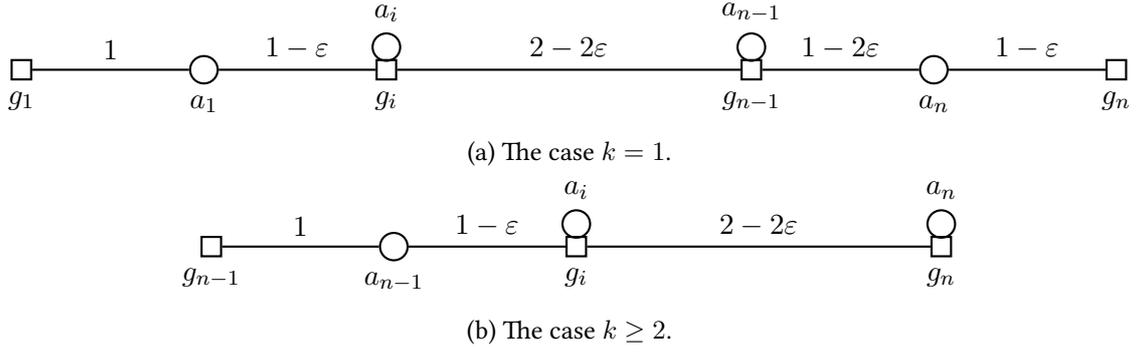
\begin{figure}[t]
\begin{subfigure}[t]{\textwidth}
    \centering
    \begin{tikzpicture}[node distance={24mm}, thick, c/.style = {draw, circle, minimum size=2mm}, s/.style = {draw, rectangle, minimum size=2mm}, t/.style = {draw, regular polygon, regular polygon sides=3, minimum size=2mm]}, invisible/.style = {minimum size=0mm, inner sep=0mm, text width=0mm, outer sep=0mm, draw=none}]
        \node[s] [label=below:$g_1$] (1) {};
        \node[c] [right of=1, label=below:$a_1$] (2) {};
        \node[s] [right of=2, label=below:$g_i$] (3) {};
        \node[c] [above of=3, node distance=3mm, label=above: $a_i$] (4) {};
        \node[s] [right of=3, right of=3, label=below:$g_{n-1}$] (5) {};
        \node[c] [above of=5, node distance=3mm, label=above: $a_{n-1}$] (6) {};
        \node[c] [right of=5, label=below:$a_n$] (7) {};
        \node[s] [right of=7, label=below:$g_n$] (8) {};
        \draw (1) -- node[midway, above] {$1$} (2);
        \draw (2) -- node[midway, above] {$1-\varepsilon$} (3);
        \draw (3) -- node[midway, above] {$2 - 2\varepsilon$} (5);
        \draw (5) -- node[midway, above] {$1 - 2\varepsilon$} (7);
        \draw (7) -- node[midway, above] {$1 - \varepsilon$} (8);
    \end{tikzpicture}
\caption{The case $k=1$.}
\label{fig:tiebreak:k=1}
\end{subfigure}
\begin{subfigure}[t]{\textwidth}
    \centering 
    \begin{tikzpicture}[node distance={24mm}, thick, c/.style = {draw, circle, minimum size=2mm}, s/.style = {draw, rectangle, minimum size=2mm}, t/.style = {draw, regular polygon, regular polygon sides=3, minimum size=2mm]}, invisible/.style = {minimum size=0mm, inner sep=0mm, text width=0mm, outer sep=0mm, draw=none}]
        \node[s] [label=below:$g_{n-1}$] (1) {};
        \node[c] [right of=1, label=below:$a_{n-1}$] (2) {};
        \node[s] [right of=2, label=below:$g_i$] (3) {};
        \node[c] [above of=3, node distance=3mm, label=above: $a_i$] (4) {};
        \node[s] [right of=3, right of=3, label=below:$g_{n}$] (5) {};
        \node[c] [above of=5, node distance=3mm, label=above: $a_{n}$] (6) {};
        \draw (1) -- node[midway, above] {$1$} (2);
        \draw (2) -- node[midway, above] {$1-\varepsilon$} (3);
        \draw (3) -- node[midway, above] {$2 - 2\varepsilon$} (5);
    \end{tikzpicture}
\caption{The case $k \geq 2$.}
\label{fig:tiebreak:kgeq2}
\end{subfigure}
\caption{The metrics used in the proof of \cref{obs:tiebreak} to give a lower bound of 5 and 7 (depending on the value of $k$) on the distortion of \cref{algname:order_match} where in Step 1 of the algorithm $a_\ell$ and $a_r$ are chosen arbitrarily among agents whose favorite item is $g_\ell$ and $g_r$, respectively, rather than being chosen to maximize the size of $\Gin$ for (a) $k=1$ and (b) $k \geq 2$. Circles correspond to agents, rectangles correspond to items, and for (a) $i \in \{2,\dots,n-2\}$ and (b) $i \in [n-2]$. The weight of each edge is the distance between its endpoints.}
\label{fig:tiebreak}
\end{figure}

\end{document}